\documentclass[11pt,letterpaper]{article}
\usepackage[margin=1in]{geometry}
\usepackage{amsmath,color,ifpdf,latexsym,graphicx,url}
\usepackage{subcaption,overpic}

\usepackage{array}
\newcolumntype{L}[1]{>{\raggedright\let\newline\\\arraybackslash\hspace{0pt}}m{#1}}
\newcolumntype{C}[1]{>{\centering\let\newline\\\arraybackslash\hspace{0pt}}m{#1}}
\usepackage[table]{xcolor}

\usepackage{tikz}
\usetikzlibrary{backgrounds, positioning, fit}
\newlength{\scaledx}
\newlength{\scaledy}
\makeatletter
\newcommand\SetScales{%
  \pgfpointxy{1}{1}
  \setlength{\scaledx}{\pgf@x};
  \setlength{\scaledy}{\pgf@y};
}
\makeatother

\graphicspath{{./figures/}{./svgtiler/}}

\def\pawn{\raisebox{-0.25em}{\includegraphics[height=1.1em]{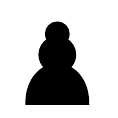}}}
\def\knight{\raisebox{-0.25em}{\includegraphics[height=1.1em]{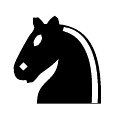}}}
\def\bishop{\raisebox{-0.25em}{\includegraphics[height=1.1em]{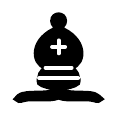}}}
\def\rook{\raisebox{-0.25em}{\includegraphics[height=1.1em]{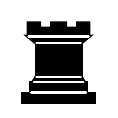}}}
\def\queen{\raisebox{-0.25em}{\includegraphics[height=1.1em]{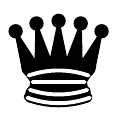}}}
\def\king{\raisebox{-0.25em}{\includegraphics[height=1.1em]{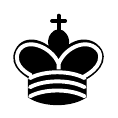}}}
\def\oneking{\king{}^1}
\def\onepawn{\pawn{}^1}
\def\oneknight{\knight{}^1}
\def\onebishop{\bishop{}^1}
\def\onerook{\rook{}^1}
\def\onequeen{\queen{}^1}

\let\realbibitem=\bibitem
\def\bibitem{\par \vspace{-1.2ex}\realbibitem}

\let\epsilon=\varepsilon
\def\defn#1{\textbf{\textit{\boldmath #1}}}

\newcommand{\solochessn}[1]{\textsc{Solo-Chess}$(#1)$}
\newcommand{\BreakText}{Split}
\newcommand{\BreakFunc}{\operatorname{\BreakText}}
\newcommand{\Break}[2]{{\BreakFunc({#1}, {#2})}}
\newcommand{\ICG}{\operatorname{ICG}}

\usepackage{hyperref}
\usepackage{amsthm, amsmath}
\theoremstyle{plain}
\newtheorem{theorem}{Theorem}[section]
\newtheorem{lemma}[theorem]{Lemma}

\newtheorem{corollary}[theorem]{Corollary}
\theoremstyle{definition}

\theoremstyle{remark}

\numberwithin{equation}{section}

\title{Complexity of Solo Chess with Unlimited Moves}

\author{%
  Josh Brunner%
    \thanks{Computer Science and Artificial Intelligence Laboratory,
      Massachusetts Institute of Technology, Cambridge, MA 02139, USA,
      \protect\url{{brunnerj,lkdc,mcoulomb,edemaine,dylanhen,tagomez7,jaysonl}@mit.edu}}
\and
  Lily Chung\footnotemark[1]
\and
  Michael Coulombe\footnotemark[1]
\and
  Erik D. Demaine\footnotemark[1]
\and
  Timothy Gomez\footnotemark[1]
\and
  Jayson Lynch\footnotemark[1]
}

\date{}

\begin{document} 
\maketitle

\begin{abstract}
  We analyze Solo Chess puzzles, where the input is an $n \times n$ board
  containing some standard Chess pieces of the same color, and the goal
  is to make a sequence of capture moves to reduce down to a single piece.
  Prior work analyzes this puzzle for a single piece type
  when each piece is limited to make at most two capture moves
  (as in the Solo Chess puzzles on chess.com).
  By contrast, we study when each piece can make an unlimited number of
  capture moves.
  We show that any single piece type can be solved in polynomial time
  in a general model of piece types,
  while any two standard Chess piece types are NP-complete.
  We also analyze the restriction (as on chess.com) that one piece type is
  unique and must be the last surviving piece, showing that in this case
  some pairs of piece types become tractable while others remain hard.
\end{abstract}

\section{Introduction}\label{Sec: Intro}
\label{sec:intro}

The classic two-player game of Chess is PSPACE-complete~\cite{storer1983complexity} or EXPTIME-complete~\cite{fraenkel1981computing} depending on whether the number of moves is limited to a polynomial.
Recent work analyzes Chess-based puzzles, including Helpmate Chess and Retrograde Chess which are PSPACE-complete \cite{brunner2020complexity} and Solo Chess which is NP-complete \cite{aravind2022chess}.
In this paper, we extend the analysis of Solo Chess to unbounded moves per piece.

First we review standard Solo Chess as implemented on chess.com
\cite{soloChess}.
All pieces are of the same color and may capture any piece except a king.
Every move must be a capture.
The objective is to find a sequence of moves (captures)
that results in only one piece remaining on the board.
If there is a king on the board, it must be last remaining piece.
Further, each piece can make a maximum of $k=2$ moves.
Past work \cite{aravind2022chess} generalizes Solo Chess to an arbitrary
limit $k$ on the number of moves per piece (and arbitrary board size),
denoting this game by (\textsc{Generalized}) \solochessn{S,k}
where $S$ is the set of allowed piece types.
They proved that \solochessn{\{\pawn\},2} and \solochessn{S,1}
can be solved in linear time,
while \solochessn{\{\bishop\},2}, \solochessn{\{\rook\},2},
and \solochessn{\{\queen\},2} are NP-complete.

This paper analyzes the complexity of Solo Chess puzzles
without the restriction on the number of moves per piece, or equivalently
when the move limit per piece is larger than the number of pieces.
We denote this game by \solochessn{S},
which is equivalent to \solochessn{S,\infty}.
We also consider the game both with and without the restriction of a single
uncapturable king, using $\oneking$ to denote the game with this restriction
and $\king$ to denote the more general case that permits
multiple capturable kings.
We also extend this notion to $T^1$ for any piece type~$T$,
meaning that there is one \defn{uncapturable} piece of type $T$
(so it must be the last piece standing).

\paragraph{Our results.}
We prove that, for any \emph{single} standard Chess piece type ($|S| = 1$),
\solochessn{S} can be solved in polynomial time.
In fact, this result holds for a very general model of piece type
defined in Section~\ref{Generalized Chess Model}.
For any \emph{two} distinct
standard Chess piece types ($|S| = 2$), neither of which are uncapturable,
we prove that \solochessn{S} is NP-complete,
by a variety of reductions.
(All problems considered here are trivially in NP.)

For the single uncapturable king $\oneking$, we prove that
the pair $S = \{\oneking,\queen\}$ can in fact be solved in polynomial time,
essentially because king moves are a subset of queen moves.
On the other hand,
$S = \{\oneking,\pawn\}, \{\oneking,\knight\}, \{\oneking,\bishop\}$
are all NP-complete.
We also give polynomial-time algorithms and NP-completeness results for
several other pairs of the form $\{S^1,T\}$ where $S^1$ is uncapturable;
see Table~\ref{tab:results} for our results restricted to standard Chess pieces.

\definecolor{header}{rgb}{0.835,0.68,0.95} 
\definecolor{gray}{rgb}{0.85,0.85,0.85}
\def\header#1{\cellcolor{header}\textbf{#1}}

\begin{table}
    \centering
    \small
    \tabcolsep=0.3\tabcolsep
    \rowcolors{2}{gray!75}{white}
    \begin{tabular}{|c|c|c|c|c|c|c|}
        \hline
        \rowcolor{header}
        ~ & \header{$\pawn$} & \header{$\king$} & \header{$\knight$} & \header{$\rook$} & \header{$\bishop$} & \header{$\queen$} \\ \hline
        \header{$\onepawn$} & P, Thm.~\ref{thm:one} & P,~\ref{thm:two} & NP-c, Thm.~\ref{thm:king-knight} & OPEN & P, Thm.~\ref{thm:two} &  P, Thm.~\ref{thm:two}  \\ 
        \header{$\oneking$} & NP-c, Thm.~\ref{thm:king-pawn} & P, Thm.~\ref{thm:one} & NP-c, Thm.~\ref{thm:king-knight} & OPEN & NP-c, Thm.~\ref{thm:king-bishop} & P, Thm.~\ref{thm:two}  \\ 
        \header{$\oneknight$} & NP-c, Thm.~\ref{thm:short range} & NP-c, Thm.~\ref{thm:short range} & P, Thm.~\ref{thm:one} & NP-c, Thm.~\ref{thm:knight-queen} & NP-c, Thm.~\ref{thm:knight-queen} & NP-c, Thm.~\ref{thm:knight-queen} \\ 
        \header{$\onerook$} & NP-c, Thm.~\ref{thm:oneRook} & NP-c, Cor.~\ref{cor:oneRookKnightKing} & NP-c, Cor.~\ref{cor:oneRookKnightKing} & P, Thm.~\ref{thm:one} & NP-c, Cor.~\ref{cor:oneRookBishop} & P, Thm.~\ref{thm:two} \\ 
        \header{$\onebishop$} & NP-c, Cor.~\ref{cor:oneBishop} & NP-c, Cor.~\ref{cor:oneBishop}  & NP-c, Cor.~\ref{cor:oneBishop}  & NP-c, Cor.~\ref{cor:oneBishopRook} & P, Thm.~\ref{thm:one} & P, Thm.~\ref{thm:two} \\ 
        \header{$\onequeen$} & NP-c, Cor.~\ref{cor:oneQueen} & NP-c, Cor.~\ref{cor:oneQueen} & NP-c, Cor.~\ref{cor:oneQueen} & NP-c, Cor.~\ref{cor:oneQueenRook} & NP-c, Cor.~\ref{cor:oneQueenBishop} & P, Thm.~\ref{thm:one} \\ \hline
    \end{tabular}
    \caption{Summary of our results for one Chess piece type, on the diagonal;
      and for two Chess piece types, one of which (the row label on the left)
      is constrained to be uncapturable and the last piece.}
    \label{tab:results}
\end{table}

This paper is divided by algorithmic vs.\ hardness results.
Section~\ref{Algorithmic Results} describes our algorithmic results for
a single piece type (Section~\ref{sec:onepieceeasy})
and for certain pairs of capturable and uncapturable piece types
(Section~\ref{sec:pairs-easy}),
which both apply to a very general model of piece type defined in
Section~\ref{Generalized Chess Model}.
Section~\ref{Hardness Results} describes our hardness results,
which fall into two main categories: reductions from Hamiltonian Path
(Section~\ref{Hamiltonian Path Reductions})
and reductions from SAT (Section~\ref{SAT Reductions}).
Section~\ref{Open Problems} concludes with some open problems.

\section{Algorithmic Results}
\label{Algorithmic Results}

In this section, we present polynomial-time algorithms for several cases of
Solo Chess.
First, in Section~\ref{Generalized Chess Model}, we define a generalized
abstract notion of moves in which pieces must always capture another piece.
Then, in Section~\ref{sec:onepieceeasy}, we show that any single piece type
in this general game (which covers all normal Chess pieces) can be solved
in polynomial time.
Finally, in Section~\ref{sec:pairs-easy}, we consider two piece types,
one of which consists of a single uncapturable piece, and show that this
problem can also be solved in polynomial time in many cases.

\subsection{Generalized Chess Model}
\label{Generalized Chess Model}

Our algorithmic results apply to a very general model of pieces moving on a
board, which includes all standard piece types from Chess,
as well as many other Fairy Chess pieces such as riders and leapers.

Define a \defn{board} to be a set \(L\) of \defn{locations}
together with a set of \defn{pieces} each assigned a unique location.
Thus, no two pieces can occupy the same location,
but some locations may be \defn{empty}.

Define a \defn{move} to be a sequence
$\langle \ell_0, \ell_1, \dots, \ell_k \rangle$ of locations.
This move is \defn{valid} (in the sense of a capture)
if $\ell_0$ and $\ell_k$ each have a piece while
$\ell_1, \dots, \ell_{k-1}$ are empty.
Executing such a valid move removes the piece at location~$\ell_k$,
and moves the piece at location $\ell_0$ to location~$\ell_k$.

Define a \defn{piece type} to be a set of moves that that type of piece
can make (when valid).  Effectively, a piece type lists, for every possible
starting location, which locations the piece can move to given that
certain other intermediate locations are empty.

For example, the Chess piece type $\rook$ is defined by the moves
\begin{align*}
\Big\{
  \big\langle (x,y), (x+1,y), \dots, (x+i,y) \big\rangle,~ &
  \big\langle (x,y), (x,y+1), \dots, (x,y+i) \big\rangle, \\
  \big\langle (x,y), (x-1,y), \dots, (x-i,y) \big\rangle,~ &
  \big\langle (x,y), (x,y-1), \dots, (x,y-i) \big\rangle ~ \Big| ~ i > 0 \,
\Big\},
\end{align*}
while the Chess piece type $\knight$ is defined by the moves
\begin{align*}
\Big\{
  \big\langle (x,y), (x + s, y + 2 t) \big\rangle,~ &
  \big\langle (x,y), (x + 2 s, y + t) \big\rangle ~ \Big| ~ s,t \in \{-1,+1\}
\Big\},
\end{align*}
where the board is either infinite or we restrict to moves whose
locations are all on the board.
Note how the $\rook$ example implements the blocking nature of rook moves ---
all locations along the way must be empty --- while the $\knight$ example
has no such blocking.
In general, if all moves are sequences of two locations,
then the piece type is \defn{nonblocking}.

We require every piece type $T$ to be \defn{closed under submoves} meaning that,
if $\langle \ell_0, \ell_1, \dots, \ell_k \rangle$ is a move in~$T$,
then $\langle \ell_i, \ell_{i+1}, \dots, \ell_j \rangle$ is a move in $T$
for any integers $0 \leq i < j \leq k$.
(This restriction is automatically satisfied by all nonblocking piece types.)

The definition of piece type supports pieces that move asymmetrically,
such as the Chess piece type~$\pawn$:
\begin{align*}
\Big\{
  \big\langle (x,y), (x + s, y + 1) \big\rangle ~ \Big| ~ s \in \{-1,+1\}
\Big\}.
\end{align*}
If a piece type is closed under reversal of the sequences representing moves,
we call it \defn{symmetric}.

Examples of Chess-like piece types not represented by this model include
the following:
\begin{enumerate}
\item The horses and elephants from Xiangqi (Chinese chess) and Janggi (Korean chess),
whose movement can be blocked by a piece they cannot capture
(so their moves are not closed under submoves);
\item The cannons from Xiangqi and Janggi,
which require a piece to jump over before performing a capture; and
\item Checkers, which land in a space past the piece captured.
\end{enumerate}

In Solo Chess, any location that is initially empty will remain empty forever,
so such locations can be omitted from the board.
More precisely, if \(\ell_i\) is initially empty,
then we can replace any move
\(\langle \ell_0, \dots, \ell_i, \dots, \ell_k \rangle\)
with
\(\langle \ell_0, \dots, \ell_{i-1}, \ell_{i+1} \dots, \ell_k \rangle\);
and any move \(\langle \ell_0, \dots, \ell_i, \dots, \ell_k \rangle\)
for which \(\ell_0\) or \(\ell_k\) is initially empty
can be deleted entirely.
These changes do not change the outcome of the Solo Chess puzzle,
so we assume henceforth that all locations are initially occupied.

\medskip
Now we define a few useful structures and prove some useful facts
about general Solo Chess puzzles and their solutions.

Define the \defn{\BreakText} operation to split a given move into a sequence
of submoves that capture any pieces along the way.  Precisely,
suppose we have a move \(m = \langle \ell_0, \ell_1, \dots, \ell_k \rangle\)
and a set \(L\) of locations including both $\ell_0$ and $\ell_k$.
Let \(\ell_{j_0}, \ell_{j_1}, \ldots, \ell_{j_{k'}}\)
be the subsequence of \(\ell_0, \ell_1, \dots, \ell_k\)
obtained by intersecting with \(L\).
Define \(\Break{m}{L}\) to be the sequence of submoves
\(\langle \ell_{j_r}, \ell_{(j_r) + 1}, \ldots,
\ell_{j_{(r+1)}} \rangle\)
for \(r = 0, 1, \dots, k'-1\) in increasing order.
If \(m\) is a move for piece type $T$ closed under submoves,
then \(\Break{m}{L}\) is a sequence of moves for piece type~$T$.
If furthermore $L$ contains all locations in $m$ currently occupied by pieces,
then $\Break{m}{L}$ is a sequence of \emph{valid} moves.

\begin{lemma}
  \label{lem:solution-no-cycle}
  Let $r$ be the location of the final piece in a valid solution.
  \begin{enumerate}
  \item For every location $v$ other than $r$, there is exactly one move out of $v$ in the solution.
  \item 
    Let $v_0, v_1, \dots$ be the sequence of locations obtained by repeatedly following the unique move out of $v_i$ in the solution, starting at $v_0$.
    Then the sequence is finite and terminates at~$r$.
  \end{enumerate}
\end{lemma}
\begin{proof}
  We prove each property separately.
  \begin{enumerate}
  \item If there is no such move, then $v$ would remain occupied forever,
    contradicting that the solution is valid.
    After the first such move, $v$ is forever empty,
    so there cannot be a second move out of $v$.
  \item 
    Suppose not.
    Because the number of locations is finite,
    the sequence must enter a cycle among locations other than~$r$.
    But then there is no move out of this cycle in the solution,
    so at least one of the locations in the cycle must remain occupied forever,
    which contradicts the validity of the solution.
    \qedhere
  \end{enumerate}
\end{proof}

\subsection{One Piece Type is Easy}
\label{sec:onepieceeasy}

In this section, we prove that Solo Chess puzzles with any single piece type
is solvable in polynomial time, even for the very general piece types
defined in Section~\ref{Generalized Chess Model}.

\begin{theorem} \label{thm:one}
  \solochessn{\{T\}} can be solved in polynomial time
  for any single piece type $T$ closed under submoves.
\end{theorem}
\begin{proof}
  Define the \defn{immediate capture graph} $\ICG$
  to be the directed graph with a vertex $\ell$ for each location $\ell$,
  and a directed edge $(\ell_0, \ell_k)$
  for every move $\langle \ell_0, \ell_1, \dots, \ell_k \rangle$
  that is valid in the initial board (without being blocked by other pieces).
  By the assumption above that all locations are occupied in the initial board,
  such moves consist of just $k=2$ locations.
  This graph can be computed in polynomial time.

  We claim that the instance is solvable if and only if
  the immediate capture graph $\ICG$
  has a \defn{spanning in-arborescence},
  i.e., a set of edges such that every vertex
  has a unique path to a common root~$r$.%
  \footnote{An \defn{(out-)arborescence} is usually defined in the
    reverse way, with a unique path from the root $r$ to every vertex.
    For this proof, we need the flipped \emph{in} version.
    Existing algorithms for out-arborescence can be applied to in-arborescence
    by reversing all edges in the graph.}
  (For symmetric piece types, the immediate capture graph is undirected,
  so it suffices to find a spanning tree and root it.)
  This property can be checked in polynomial time via Edmonds' 1967 Algorithm
  \cite{edmonds-arborescence},
  or its $O(|E| + |V| \log |V|)$ optimization \cite{fast-arborescence};
  or in $O(|E| + |V|)$ time using a modified depth-first search
  \cite[Exercise 6.10]{arborescences}.

  If a spanning in-arborescence exists, we can solve the instance as follows.
  Every in-arborescence with more than one vertex has a \defn{leaf} vertex,
  i.e., a vertex with no incoming edges and one outgoing edge.
  Repeatedly find a leaf and make the corresponding move from the
  leaf to its \defn{parent} (the vertex reached via the one outgoing edge),
  deleting the leaf.
  Every move is valid because it corresponds to an edge in the
  immediate capture graph and, by construction,
  the two relevant pieces have not yet been captured.
  Because the in-arborescence is spanning, we reduce to a single piece
  at location $r$ in the end.

  Conversely, suppose that the instance is solvable via a sequence of moves.
  Let $r$ be the location of the final piece in the solution.
  We will show that $\ICG$ admits a spanning in-arborescence rooted at~$r$.
  It suffices to show that there is a directed walk in $\ICG$
  from each location to~$r$ \cite[Theorem 2.5(d)]{arborescences}.

  Let $v_0$ be any location and let $v_0, v_1, \dots, r$ be the sequence of locations from Lemma~\ref{lem:solution-no-cycle}.
  For each pair $(v_i, v_{i+1})$ of locations, there is a move \(m\) from \(v_i\) to \(v_{i+1}\) in the solution.
  Then \(\Break{m}{L}\) (where \(L\) is the set of all locations)
  is a sequence of two-location moves,
  each of which by definition corresponds to an edge of~\(\ICG\).
  Hence we obtain a walk from $v_i$ to $v_{i+1}$ in~$\ICG$.
  By concatenating these walks, we obtain a walk from $v_0$ to $r$ in~$\ICG$.
  Therefore $\ICG$ admits a spanning in-arborescence.
\end{proof}

\subsection{One Hero and Villains are Easy}
\label{sec:pairs-easy}

In this section, we present an algorithm for solving certain Solo Chess puzzles
consisting of $n$ copies of one piece type,
and a single copy of a second piece type
which is required to be the final piece on the board
(and is thus uncapturable).
This is a generalization of the one-king restriction from Solo Chess
which says that, if there is a king in the initial puzzle,
it must be the final piece on the board in a valid solution.

We call the single uncapturable piece the \defn{hero},
and the pieces that we have $n$ copies \defn{villains}.
Let $S$ be the hero piece type, and $T$ be the villain piece type.
We require that $S \subseteq T$,
that is, every move for a hero is also a move for a villain.
We also require that $T$ is symmetric.
Two useful cases to think about are when the hero is a Chess king or pawn,
and the villain is a Chess queen.
Among Chess pieces, $\{S^1, T\}$ includes
$\{\oneking,\queen\}$, $\{\onepawn,\queen\}$, $\{\onepawn,\bishop\}$,
$\{\onerook,\queen\}$, and $\{\onebishop,\queen\}$.

The basic idea of the algorithm is to use villains capturing villains
to collapse the pieces down onto a path which the hero can traverse,
thus capturing every villain.
The main difficulty with this approach is \emph{blocking}:
the hero itself prevents villains from moving through it,
which may require moves to be delayed until after the hero is out of the way.

\subsubsection{Walks in Immediate Capture Graphs}

Let \(V\) and \(B\) be disjoint sets of locations.
Define the \defn{immediate capture graph \(\ICG(V, B)\)} to be
the directed graph whose vertex set is \(V\),
and which has an edge \((\ell_0, \ell_m)\) whenever
\(\langle \ell_0, \ell_1, \ldots, \ell_m \rangle \in T\)
is a villain move such that none of \(\ell_1, \ell_2, \ldots, \ell_{m-1}\) are in \(V \cup B\).
(Note that \(\ell_1, \ell_2, \ldots, \ell_{m-1}\) may not be in \(V\).)
This is the graph of valid villain moves when
villains are placed on the locations in \(V\)
and immobile uncapturable blocking pieces are placed on the locations in~\(B\).

We now prove some useful facts ensuring the existence of (directed) walks
in \(\ICG(V, B)\).

\begin{lemma}
  \label{lem:break-icg}
  Let \(V\) and \(B\) be disjoint sets of locations.
  Suppose \(\langle \ell_0, \ell_1, \ldots, \ell_m \rangle \in T\)
  is a villain move such that
  \(\ell_0,\ell_m \in V\) and
  \(\{\ell_0, \ell_1, \dots, \ell_m\} \cap B = \emptyset\).
  Then there is a directed walk in \(\ICG(V, B)\) from \(\ell_0\) to \(\ell_m\).
\end{lemma}
\begin{proof}
  Because \(\ell_0,\ell_m \in V\),
  \(\ell_0\) and \(\ell_m\) are both vertices of~\(\ICG(V, B)\).
  Consider the sequence of villain moves \(\Break{\langle \ell_0, \ell_1, \ldots, \ell_m \rangle}{V}\).
  By definition of $\BreakFunc$,
  each such move \(\langle \ell_s, \ell_{s+1}, \ldots \ell_{s+r} \rangle\)
  has the property that
  \(\{\ell_{s+1}, \ell_{s+2}, \ldots, \ell_{s+r-1}\} \cap (V \cup B) = \emptyset\)
  while
  \(\ell_s, \ell_{s+r} \in V\).
  Thus \((\ell_s, \ell_{s+r})\) is an edge of~\(\ICG(V, B)\).
  By concatenating these edges together,
  we obtain a directed walk in \(\ICG(V, B)\) from \(\ell_0\) to \(\ell_m\).
\end{proof}
\begin{corollary}
  \label{cor:break-icg-walk}  
  Let \(V\) and \(B\) be disjoint sets of locations,
  and let \(m_0, m_1, \ldots, m_k \in T\) be a sequence of villain moves
  such that
  \begin{enumerate}
  \item the destination of \(m_i\) is the source of \(m_{i+1}\);
  \item each move's source and destination is in \(V\); and
  \item no move passes through a location in \(B\).
  \end{enumerate}
  Then there is a directed walk in \(\ICG(V, B)\) from the source of \(m_0\) to the destination of \(m_k\).
\end{corollary}
\begin{proof}
  Each move satisfies the conditions of Lemma~\ref{lem:break-icg},
  and so there is a directed walk in \(\ICG(V, B)\)
  from the source to the destination of each move.
  Concatenating these walks together gives the desired walk.
\end{proof}
\begin{corollary}
  \label{cor:icg-walk-subset}
  Let \(V\) and \(B\) be disjoint sets,
  and let \(V'\) and \(B'\) be disjoint sets
  such that \(V \subseteq V'\) and \(B' \subseteq B\).
  (Note the asymmetry.)
  Suppose there is a directed walk from \(\ell_0\) to \(\ell_1\) in \(\ICG(V, B)\).
  Then there is also a directed walk from \(\ell_0\) to \(\ell_1\) in \(\ICG(V', B')\).
\end{corollary}
\begin{proof}
  A walk from \(\ell_0\) to \(\ell_1\) in \(\ICG(V, B)\)
  is a sequence of moves, each of whose source and destination is in \(V \subseteq V'\)
  and none of which pass through locations in \(B \supseteq B'\).
  Thus, by Corollary~\ref{cor:break-icg-walk},
  there is a directed walk from \(\ell_0\) to \(\ell_1\) in \(\ICG(V', B')\).
\end{proof}

\subsubsection{Hero Paths}

Given a board, a \defn{hero location sequence} is any sequence of locations
$p_0, p_1, \ldots, p_k$ where \(p_0\) is the hero's initial position
in the given board.
A \defn{hero path} is a hero location sequence
with the additional property that the hero has a
sequence of valid moves $\langle p_0, \dots, p_1 \rangle$,
$\langle p_1, \dots, p_2 \rangle$, \dots,
$\langle p_{k-1}, \dots, p_k \rangle$
in the initial board (without any other moves being made).
This definition is particularly simple for e.g.\ Chess kings or pawns,
but more generally, hero moves could be blocked by the villain pieces.
We will prove below in Lemma~\ref{lem:hero path} that it suffices to
look at solutions where the hero follows a hero path,
but for now we know that the hero at least follows a hero location sequence.

For a fixed hero location sequence $p_0, p_1, \ldots, p_k$
(collectively denoted~$p$)
and an index $i$ with $0 \leq i < k$,
we define \(H_i(p) = \ICG(L \setminus \{p_0, p_1, \ldots, p_i\}, \{p_i\})\) to be the immediate capture graph
(as defined above)
of the villains on the board after making just the first \(i\) hero moves
(so that the hero is now at \(p_i\) and \(p_0, p_1, \ldots, p_{i-1}\) are empty).
We will just write \(H_i\) (omitting the hero location sequence~\(p\))
when it is clear from context.

Define a villain at location \(v\) to be \defn{strongly capturable}
by hero location sequence $p_0, p_1, \ldots, p_k$ if,
for some index $i$ with $0 \leq i < k$,
there is a directed walk in \(H_i\) from \(v\) to \(p_{i+1}\).
Intuitively, when the hero is at location $p_i$,
the villain can move to hero location $p_{i+1}$,
and then get captured by the hero's next move.
By definition, all of \(p_1, p_2, \ldots, p_k\) are strongly capturable.

\begin{lemma}
  \label{lem:strongly-capturable-later}
  Let $p_0, p_1, \ldots, p_k$ be a hero location sequence.
  Suppose there is a directed walk in \(H_i\)
  from \(v\) to \(p_j\), with \(i < j\).
  Then \(v\) is strongly capturable by $p_0, p_1, \ldots, p_k$.
\end{lemma}
\begin{proof}
  Without loss of generality we can assume the walk is minimal,
  so that \(p_j\) is the first place the walk visits any of \(\{p_{i+1}, p_{i+2}, \ldots, p_k\}\).
  Then this walk forms a sequence of moves each of whose source and destination
  is in \(L \setminus \{p_0, p_1, \ldots, p_{j-1}\}\) and which do not pass through \(p_{j-1}\) (because they are edges of \(H_i\)).
  By Corollary~\ref{cor:break-icg-walk},
  there is a directed walk from \(v\) to \(p_j\) in \(H_{j-1}\),
  and so \(v\) is strongly capturable by
  $p_0, p_1, \ldots, p_k$ using index~$j-1$.
\end{proof}

Next we show that knowing the sequence of locations the hero visits
suffices to solve the puzzle.
That is, given a hero path $p_0, p_1, \ldots, p_k$,
we will show how to determine in polynomial time
whether there a valid solution such that the hero makes exactly $k$ moves
through this sequence of locations.
Specifically, we show that it is a necessary and sufficient condition for
all villains to be strongly capturable by $p_0, p_1, \ldots, p_k$,
in two parts.

\begin{lemma}
  \label{lem:soln-strongly-capturable}
  Suppose there is a solution such that the hero makes exactly $k$ moves
  along the hero location sequence $p_0, p_1, \ldots, p_k$.
  Then all villains are strongly capturable by $p_0, p_1, \ldots, p_k$.
\end{lemma}
\begin{proof}
  Consider a villain at location \(v_0\).
  Let \(v_0, v_1, \dots, p_k\) be the sequence of locations from
  Lemma~\ref{lem:solution-no-cycle} applied to the solution.
  Truncate the sequence at the first location \(v_r = p_{i+1}\)
  where \(v_r \in \{p_0, p_1, \ldots, p_k\}\),
  so that in particular \(v_j \notin \{p_0, p_1, \ldots, p_i\}\)
  for \(0 \leq j \le r\).

  For each pair \((v_j, v_{j+1})\) with \(0 \le j < r\), there is a villain move $\langle v_j = \ell_0, \ell_1, \ldots, \ell_m = v_{j+1} \rangle$ in the solution.
  These villain moves must occur in the solution before the hero move from \(p_i\) to \(p_{i+1}\), because the hero must make the last move to \(p_{i+1}\).
  Thus \(p_i \notin \{\ell_0, \ell_1, \ldots, \ell_m\}\)
  as \(p_i\) is still occupied when each villain move occurs.

  Therefore we have a sequence of moves
  each of whose source and destination is in \(L \setminus \{p_0, p_1, \ldots, p_i\}\)
  and none of which passes through \(p_i\).
  By Corollary~\ref{cor:break-icg-walk},
  there is a directed walk from \(v_0\) to \(p_{i+1}\) in \(H_i\),
  and so \(v_0\) is strongly capturable.
\end{proof}

\begin{lemma}
  \label{lem:strongly-capturable-soln}
  Let \(p_0, p_1, \ldots, p_k\) be a hero path
  such that every villain is strongly capturable by \(p_0, p_1, \ldots, p_k\).
  Then there is a solution such that the hero makes exactly $k$ moves
  through the sequence of locations $p_0, p_1, \ldots, p_k$.
\end{lemma}
\begin{proof}
  For every villain starting at location~$v$,
  by the definition of strongly capturable,
  we obtain an index \(i\) with $0 \leq i < k$,
  which we call the villain's \defn{rank},
  and a directed walk in \(H_i\) from \(v\) to \(p_{i+1}\).
  By removing any cycles in the walk, we can assume that the walk
  is in fact a simple path.

  We will output in reverse order a sequence of moves
  that solves the instance.
  Sort the villains by rank, breaking ties arbitrarily,
  and process the villains in order from highest rank to lowest.
  For each index \(i\) with \(0 \leq i < k\),
  add the hero move from $p_i$ to $p_{i+1}$ to the output,
  and process the villains of rank $i$ as follows.
  For each villain starting at location $v$ and having rank \(i\),
  find the earliest location $v'$ in the simple path
  from $v$ to $p_{i+1}$ such that a move to or from $v'$
  is already in our output list of moves.
  Such a location always exists because we just added a hero move to $p_{i+1}$.
  Add (in reverse order) the sequence of villain moves corresponding to the
  subpath from $v$ to~$v'$.
  (In particular, we do not add any moves if $v=v'$.)
  Define each villain move added by processing a villain of rank \(i\)
  to also have rank~\(i\).

  We claim that there are no villain moves out of \(p_0, p_1, \ldots, p_k\)
  and that no villain move of rank \(i\) captures~\(p_i\).
  This is because the path from \(v\) to \(p_{i+1}\)
  does not include any of \(p_0, p_1, \ldots, p_i\),
  because those are not vertices of~\(H_i\);
  and because there are already hero moves
  involving all of \(p_i, p_{i+1}, \ldots, p_k\).

  We must show that the resulting sequence of moves is a solution
  that uses the hero path $p_0, p_1, \ldots, p_k$.
  The full hero path \(p_0, p_1, \ldots, p_k\) appears in the solution
  by construction.
  Every location except $p_k$ appears exactly once as the start of a move
  in the output, and there is no output move out of~\(p_k\).
  Each villain move of rank \(i\) occurs when the hero is at \(p_i\),
  having captured all of \(p_0, p_1, \ldots, p_i\).
  No villain move captures the hero because no villain move of rank \(i\)
  captures~\(p_i\).

  It remains to show that every move is legal at the time it is made.
  There are two ways a move could have been illegal:
  either it was made to or from a location that is now empty,
  or it was blocked by an intervening piece.
  The source of a move cannot be empty
  because each location occurs as the source of a move at most once.
  The destination of a move cannot be empty
  because every move is either to $p_k$ or to a location
  for which the unique move from that location occurs later in time.
  No hero move is blocked by the definition of hero path:
  \(p_0, p_1, \ldots, p_k\) must be a legal sequence of moves
  even without moving the villains.
  No villain move is blocked because villain moves of rank \(i\)
  are made only along edges of~\(H_i\);
  by definition of \(H_i\), such a move can be blocked only by
  \(p_0, p_1, \ldots, p_{i-1}\), which are empty when the hero is at~$p_i$.
\end{proof}

Note that Lemma~\ref{lem:soln-strongly-capturable}
applies to any hero location sequence \(p_0 \ldots p_k\),
but Lemma~\ref{lem:strongly-capturable-soln} requires a hero path.
There is a technical issue here, because
the hero moves might be blocked in a hero location sequence,
which a solution might avoid by moving villains out of the way.
However, we can show that there is always an alternate solution
that avoids doing so:

\begin{lemma} \label{lem:hero path}
  Suppose a puzzle has a solution.
  Then there is a solution such that the sequence of locations
  visited by the hero forms a hero path.
\end{lemma}
\begin{proof}
  Consider taking just the subsequence of hero moves from the solution,
  and attempting to play them without making any villain moves.
  When doing so, we may encounter a hero move \(m\)
  that is illegal because it is blocked by some set \(V_m\) of villains
  which have not yet been captured.
  Replace each such hero move \(m\) with \(\Break{m}{V_m}\),
  which is a sequence of valid moves.
  The resulting sequence of hero moves forms a hero path
  \(p_0, p_1, \ldots, p_k\),
  which contains the original hero location sequence
  and also contains all of the \(V_m\) sets.

  It remains to show that there exists a solution such that
  the hero makes exactly $k$ moves following this hero path.
  By Lemma~\ref{lem:strongly-capturable-soln}, it suffices to show that
  every villain is strongly capturable by \(p_0, p_1, \ldots, p_k\).
  We do this using a similar argument to Lemma~\ref{lem:soln-strongly-capturable}.

  Consider a villain at location \(v_0\).
  Let \(v_0, v_1, \dots, v_r\) be the sequence of locations from
  Lemma~\ref{lem:solution-no-cycle} applied to the original solution.
  Truncate the sequence at the first location \(v_r = p_{i+1}\)
  where \(v_r \in \{p_0, p_1, \ldots, p_k\}\),
  so that in particular \(v_j \notin \{p_0, p_1, \ldots, p_i\}\)
  for \(0 \leq j \le r\).

  Let \(p_s\) be the location of the hero in the original solution
  when the villain move from \(v_{r-1}\) to \(p_{i+1}\) gets made,
  so that all of the above villain moves occur in the solution
  before the unique hero move out of \(p_s\).
  None of these villain moves passes through \(p_s\)
  because it is occupied when they occur.
  It must be that \(s < i + 1\),
  for in the solution the hero cannot move into or through \(p_{i+1}\)
  until after the last villain move to \(p_{i+1}\),
  which occurs after the hero move to \(p_s\).

  By Corollary~\ref{cor:break-icg-walk},
  there is a directed walk from \(v_0\) to \(p_{i+1}\) in \(H_s(p)\),
  and so by Lemma~\ref{lem:strongly-capturable-later},
  \(v_0\) is strongly capturable by \(p_0, p_1, \ldots, p_k\).
\end{proof}

By the above lemmas, we can consider a solution to the puzzle to consist of
just a hero path for which all villains strongly capturable.

Next we define the notion of ``weak capturability'',
which can be used to rule out prefixes of solution hero paths.
For a hero path \(p_0, p_1, \ldots, p_k\),
define \(G_i(p) = ICG(L \setminus \{p_0, p_1, \ldots, p_i\}, \emptyset)\)
to be the immediate capture graph of the board if we remove
\(p_0, p_1, \ldots, p_i\) entirely.
Define a villain at location \(v\) to be
\defn{weakly capturable after \(p_0, p_1, \ldots, p_k\)} if
there is a hero path \(p_0, p_1, \ldots, p_{k'}\)
with \(p_0, p_1, \ldots, p_k\) as a prefix,
such that there is a directed walk in \(G_k\) from \(v\) to \(p_{k'}\).
While $H_i(p)$ is a subgraph of $G_i(p)$, weak capturability does not
necessarily imply strong capturability with the same hero path
because only the latter gets to pick an index~$i$;
weak capturability must use $i=k$.
(Intuitively, strong capturability means that the villain has effectively
already been captured, while weak capturability means that the villain
could be in the future.)

We show that hero paths can be ruled out as potential prefixes of
solutions if they do not at least make all the villains weakly capturable:

\begin{lemma}
  \label{lem:strong-or-weak}
  Let \(p_0, p_1, \dots, p_k\) be a hero path, and suppose that there is a
  solution whose hero path has \(p_0, p_1, \ldots, p_k\) as a prefix.
  Then every villain is either strongly capturable by \(p_0, p_1, \ldots, p_k\)
  or weakly capturable after \(p_0, p_1, \ldots, p_k\).
\end{lemma}
\begin{proof}
  Let \(p_0, p_1, \ldots, p_{k'}\) be the hero path in the solution,
  and let \(v\) be the location of a villain.
  By Lemma~\ref{lem:soln-strongly-capturable},
  \(v\) is strongly capturable by \(p_0 \ldots p_{k'}\).
  That is, there exists an index \(i < k'\) and a directed walk in \(H_i\)
  from \(v\) to \(p_{i+1}\).
  If \(i < k\), then \(v\) is strongly capturable by \(p_0 \ldots p_k\),
  so suppose \(i \ge k\).

  By Corollary~\ref{cor:icg-walk-subset}, a directed walk in
  \(H_i = \ICG(L \setminus \{p_0, p_1, \ldots, p_i\}, \{p_i\})\)
  extends to a directed walk in
  \(G_k = \ICG(L \setminus \{p_0, p_1, \ldots, p_k\}, \emptyset)\).
  Thus there is a directed walk from \(v\) to \(p_{i+1}\) in \(G_k\),
  and so \(v\) is weakly capturable after \(p_0, p_1, \ldots, p_k\).
\end{proof}

\subsection{Interesting Hero Paths}

Consider a villain at location $v$.
Define a hero path $p_0 \dots p_k$ to be \defn{$v$-interesting} if
\begin{enumerate}
\item $v$ is not strongly capturable by \(p_0, p_1, \ldots, p_k\); and
\item $v$ is weakly capturable after \(p_0, p_1, \ldots, p_k\).
\end{enumerate}

We will show that $v$-interesting paths form a tree for any~$v$,
and that hero paths have $v$-interesting prefixes for some~$v$.
Together these limit the number of possible hero paths we need to search.

\begin{lemma}
  \label{lem:interesting-prefix}
  If $p_0, p_1, \ldots, p_k$ is a $v$-interesting path,
  then every prefix of the path is also.
\end{lemma}
\begin{proof}
  Let \(p_0, p_1, \ldots, p_j\) be a prefix of \(p_0, p_1, \ldots, p_k\),
  i.e., $j \leq k$.

  If $v$ is strongly capturable by \(p_0 \ldots p_j\),
  then it is strongly capturable by \(p_0 \ldots p_k\) also,
  using the same index.

  Suppose $v$ is weakly capturable after $p_0, p_1, \ldots, p_k$.
  Then there is some hero path \(p_0, p_1, \ldots, p_{k'}\)
  with \(p_0, p_1, \ldots, p_k\) as a prefix
  and a directed walk from \(v\) to \(p_{k'}\)
  in \(G_k = \ICG(L \setminus \{p_0, p_1, \ldots, p_k\}, \emptyset)\).
  By Corollary~\ref{cor:icg-walk-subset},
  this directed walk extends to a directed walk in
  \(G_j = \ICG(L \setminus \{p_0, p_1, \ldots, p_j\}, \emptyset)\).
  Thus \(v\) is weakly capturable after \(p_0, p_1, \ldots, p_j\).
\end{proof}

\begin{corollary}
  \label{cor:minimal-interesting}
  Suppose \(p_0, p_1, \ldots, p_k\) with \(k > 0\)
  is a \defn{minimal solution}; that is, no prefix is also a solution.
  Then there is a villain at location \(v\)
  such that \(p_0, p_1, \ldots, p_j\) is \(v\)-interesting for all \(j < k\).
\end{corollary}
\begin{proof}
  Because \(p_0, p_1, \ldots, p_{k-1}\) is not a solution,
  by Lemma~\ref{lem:strongly-capturable-soln}, there is some \(v\)
  that is not strongly capturable by \(p_0, p_1, \ldots, p_{k-1}\).
  By Lemma~\ref{lem:strong-or-weak},
  \(v\) is weakly capturable after \(p_0, p_1, \ldots, p_{k-1}\).
  Hence \(p_0, p_1, \ldots, p_{k-1}\) is \(v\)-interesting.
  Finally, by Lemma~\ref{lem:interesting-prefix},
  \(p_0, p_1, \ldots, p_j\) is \(v\)-interesting for all \(j < k\).
\end{proof}

\begin{lemma}
  \label{lem:weak-connected}
  Suppose the villain piece type \(T\) is symmetric.
  If \(v\) is weakly capturable after a hero path $p_0, p_1, \ldots, p_k$,
  then \(v\) is connected to \(p_k\) in \(G_{k-1}\).
\end{lemma}
\begin{proof}
  By definition of weakly capturable, there is a hero path
  \(p_0, p_1, \ldots, p_k, p_{k+1}, \ldots, p_{k'}\) such that
  there is a path \(P_1\) from \(v\) to \(p_{k'}\) in \(G_k\).
  By Corollary~\ref{cor:icg-walk-subset},
  \(P_1\) extends to a path \(P'_1\) in \(G_{k-1}\).

  Consider the suffix of hero moves \(p_k, p_{k+1}, \ldots, p_{k'}\)
  from the hero path.
  Because the hero piece type \(S\)
  is a subset of the villain piece type \(T\),
  we know by Corollary~\ref{cor:break-icg-walk} that
  this sequence of hero moves extends to a path \(P_2\) from \(p_k\) to \(p_{k'}\)
  in \(\ICG(\{p_k, p_{k+1}, \ldots, p_{k'}\}, \emptyset)\).
  By Corollary~\ref{cor:icg-walk-subset}, \(P_2\) extends to a path \(P_2'\) in \(G_{k-1}\).

  Concatenating \(P'_1\) with the reverse of \(P'_2\)
  (by symmetry of villain moves),
  we obtain a path from \(v\) to \(p_k\) in \(G_{k-1}\).
\end{proof}

\begin{lemma}
  \label{lem:no-interesting-cycle}
  Suppose the villain piece type \(T\) is symmetric.
  If $p_0, p_1, \ldots, p_k$ and $q_0, q_1, \ldots, q_j$
  are two different hero paths with the same start and end points,
  then at most one of them is $v$-interesting.
\end{lemma}
\begin{proof}
  Suppose for contradiction that both paths are $v$-interesting.
  Let $i+1$ be the first index at which the two paths diverge, so $p_i = q_i$ but $p_{i+1} \ne q_{i+1}$.
  By Lemma~\ref{lem:interesting-prefix},
  we can assume without loss of generality that \(k\) and \(j\) are the
  smallest integers $> i$ for which $p_k = q_j$.
  By Lemma~\ref{lem:interesting-prefix},
  \(p_0, p_1, \ldots, p_{i+1}\) is \(v\)-interesting.
  By Lemma~\ref{lem:weak-connected},
  there is a path \(P_1\) in \(G_i\) from \(v\) to \(p_{i+1}\).
  Truncate \(P_1\) at the first point that it enters
  \(R = \{p_{i+1}, p_{i+2}, \ldots, p_k, q_{i+1}, q_{i+2}, \ldots, q_j\}\);
  without loss of generality, suppose that this point is \(q_r\)
  for some \(i + 1 \le r \le j\).
  Thus \(P_1\) becomes a path in \(\ICG\big((L \setminus (p \cup q)) \cup \{q_r\}, R \setminus \{q_r\}\big)\).
  By Corollary~\ref{cor:icg-walk-subset}, \(P_1\) extends to a path \(P'_1\) in \(H_{k-1}(p)\).

  Consider the suffix of hero moves \(q_r, q_{r+1}, \ldots, q_j\)
  from the corresponding hero path.
  By the definition of hero path and minimality of~\(j\),
  none of these moves pass through any of \(p_{i+1}, p_{i+2}, \ldots, p_{k-1}\).
  Because the hero move type \(S\) is a subset of the villain piece type~\(T\),
  we know by Corollary~\ref{cor:break-icg-walk} that
  this sequence of hero moves extends to a path \(P_2\) from \(q_r\) to \(q_j\)
  in \(\ICG(\{q_r, q_{r+1}, \ldots, q_j\}, \{p_{i+1}, p_{i+2}, \ldots, p_{k-1}\})\).
  By Corollary~\ref{cor:icg-walk-subset}, \(P_2\) extends to a path \(P_2'\) in \(H_{k-1}(p)\).

  Finally, concatenating \(P_1'\) and \(P_2'\),
  we obtain a path in \(H_{k-1}(p)\) from \(v\) to \(q_j = p_k\).
  But then \(v\) is strongly capturable by \(p_0, p_1, \ldots, p_k\),
  contradicting that it was \(v\)-interesting.
\end{proof}

Call a hero path \defn{interesting} if it is \(v\)-interesting
for some villain location \(v\).

\begin{lemma}
  \label{lem:compute-interesting}
  Given a hero path \(p_0, p_1, \ldots, p_k\), we can compute
  in polynomial time whether it is interesting and whether it is a solution.
\end{lemma}
\begin{proof}
  We can compute whether a villain location \(v\) is strongly capturable
  by constructing the graphs \(H_i\) and
  checking whether \(p_{i+1}\) is reachable from \(Q\) for each~$i$.

  We can also compute whether a villain location \(v\) is weakly capturable
  by constructing the graph \(G_k\) and searching over all locations
  for a location \(\ell\) such that \(\ell\) is reachable from \(v\) in \(G_k\)
  and also reachable from \(p_k\) in the graph of possible hero moves.
\end{proof}

\begin{theorem} \label{thm:two}
  \solochessn{\{S^1,T\}} can be solved in polynomial time
  for any two piece types $S \subseteq T$ closed under submoves
  where \(T\) is symmetric.
\end{theorem}
\begin{proof}
  By Corollary~\ref{cor:minimal-interesting},
  every prefix of a minimal solution is necessarily interesting.
  Furthermore, by Lemma~\ref{lem:no-interesting-cycle},
  there is at most one interesting hero path ending at each location.

  Our algorithm finds a minimal solution
  by searching over the tree of all interesting hero paths.
  By Lemma~\ref{lem:compute-interesting},
  we can test whether a hero path is interesting in polynomial time,
  as well as whether it solves the instance.
  There are only polynomially many locations, interesting paths,
  and moves to try, so this search takes polynomial time.
\end{proof}

\section{Hardness Results}
\label{Hardness Results}

In this section, we prove that \solochessn{S} is NP-complete
for any set $S$ of two distinct standard Chess pieces.
These reductions also work when one of the piece types is denoted
\defn{special}, restricting that there is only one copy of that piece
and that it must be the last piece on the board.
This is a generalization of the one-king restriction $\oneking$;
we use the same notation $T^1$ for other piece types~$T$.
But all of our hardness reductions also work when piece type $T$
is not constrained and could be captured.

We divide the section into two subsections based on the source of reduction.
In Section~\ref{Hamiltonian Path Reductions},
we give multiple reductions from Hamiltonian Path,
where the special piece needs to visit the other pieces.
In Section~\ref{SAT Reductions}, we reduce from a special case of 3SAT,
where the special piece sets variables and satisfies clauses.

\subsection{Hamiltonian Path Reductions}
\label{Hamiltonian Path Reductions}

All of these reductions are from Hamiltonian Path
in maximum-degree-3 grid graphs with a specified start vertex and
possibly a specified end vertex, each of degree~$1$,
or generalizations thereof
(e.g., sometimes we do not need the grid-graph or degree-$1$ property).
This problem is NP-hard by a slight modification to \cite{Degree3GridHamPath}
described in Appendix~\ref{app:start vertex}.


The first reduction, when the special piece is a knight,
is particularly easy:

\begin{theorem} \label{thm:short range}
  For any $T \in \{\pawn,\king\}$,
	\solochessn{\{\oneknight,T\}} and \solochessn{\{\knight,T\}} are NP-hard.
\end{theorem}
\begin{proof}
  The reduction is from Hamiltonian Path in grid graphs with a given start vertex $s$ (Lemma~\ref{lem:start vertex}).
  Figure~\ref{fig:PawnKnightNP} shows an example of the reduction. We rotate the grid graph by $\arctan \frac{1}{2}$ and scale it so that adjacent vertices form valid knight moves. We place pawns or kings at the grid-graph vertices, except for $s$ where we place a knight (the only knight in the construction). The pawns or kings cannot make any captures, so they are immobile. Thus the knight must capture all of the other pieces; such a sequence of captures corresponds to a Hamiltonian path starting at~$s$.
\end{proof}

\begin{figure}
	\centering
	\subcaptionbox{Pawns and one knight}{\includegraphics[scale=0.5]{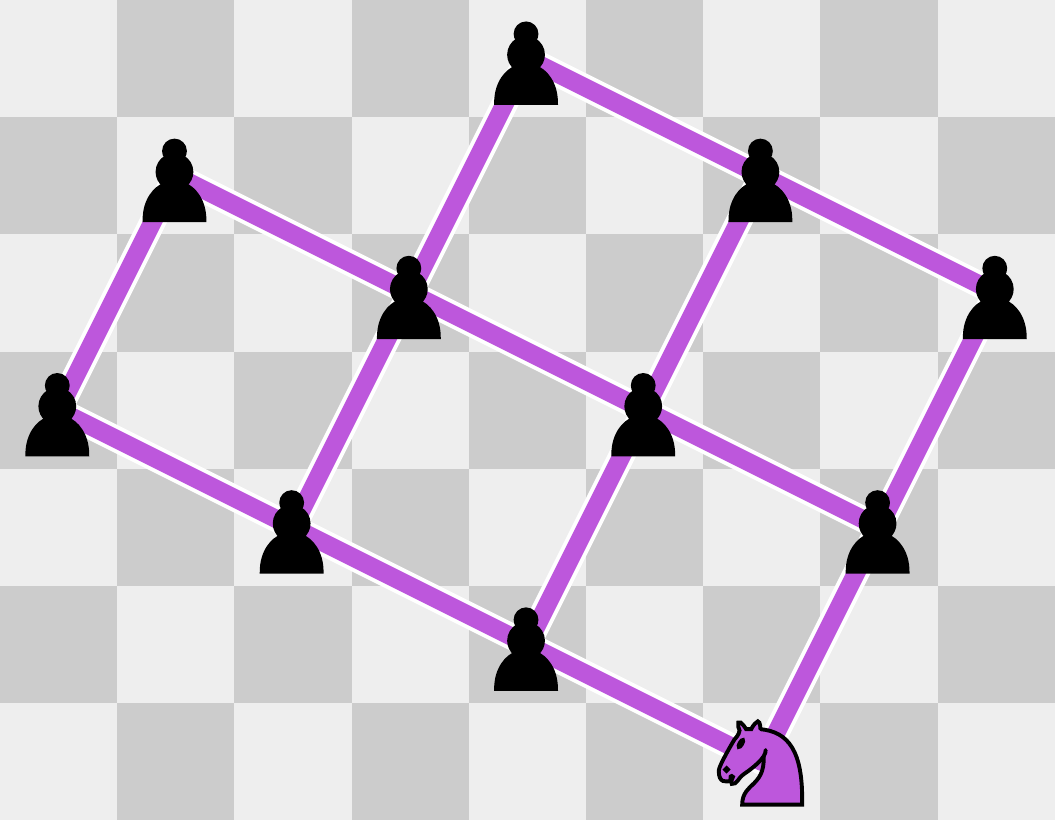}}\hfil
	\subcaptionbox{Kings and one knight}{\includegraphics[scale=0.5]{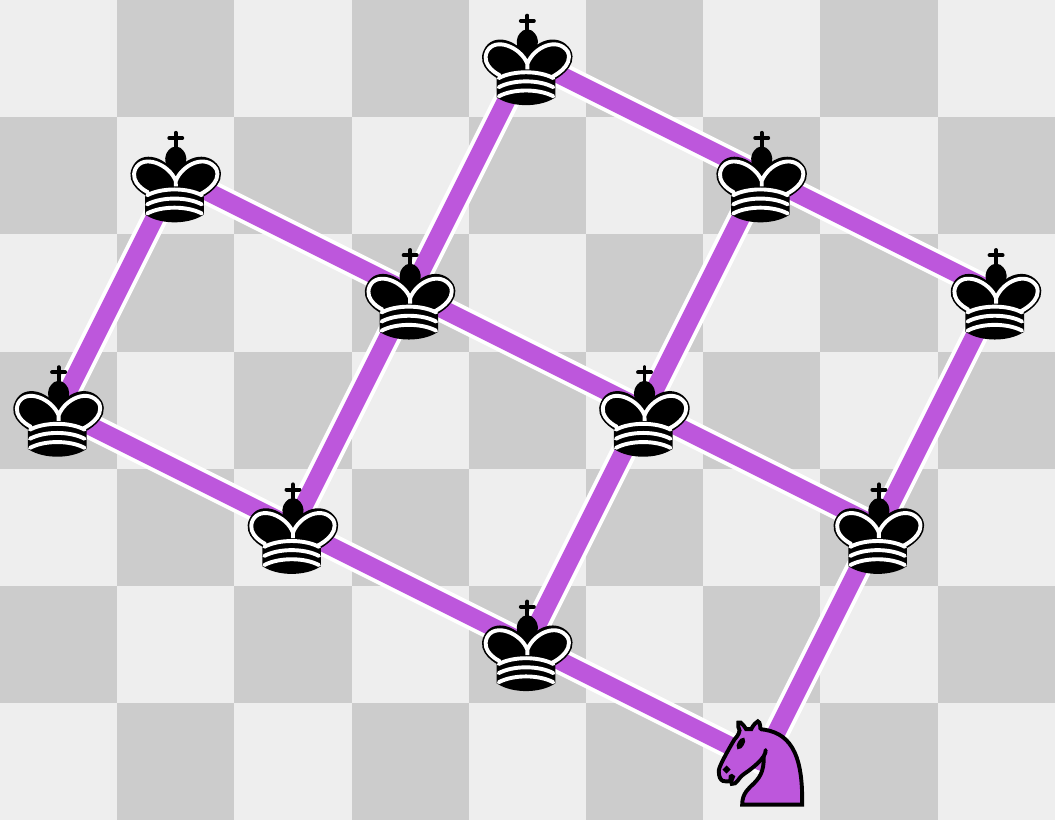}}
  \caption{Placing pawns or kings and one knight to simulate Hamiltonian Path in a grid graph with a specified start vertex.}
	\label{fig:PawnKnightNP}
\end{figure}

Next we give a reduction from Hamiltonian Path to \solochessn{\{\onerook, \pawn\}}.
This reduction forms the basis for proving NP-hardness of several
other piece combinations by scaling and/or rotation.
Like the previous reduction,
the main idea behind these constructions is to place many pieces of one type
so that they have no available moves, and a single piece of the special type
which must then capture all the other pieces, requiring a Hamiltonian path.

\begin{theorem}\label{thm:oneRook}
  \solochessn{\{\onerook, \pawn\}} and \solochessn{\{\rook,\pawn\}} are NP-hard.
\end{theorem}
\begin{proof}
  The reduction is from Hamiltonian Path in a maximum-degree-$3$ graph
  with a given start vertex $s$ and destination vertex $t$,
  both of degree~$1$ (Lemma~\ref{lem:start vertex}).

  Refer to Figure~\ref{fig:PawnRookNP}.
  Each vertex other than \(s\) and \(t\) consists of seven pawns arranged among two rows,
  with four pawns marking the corners of a very-wide height-$2$ rectangle,
  and three pawns on the bottom row of the rectangle
  which each form half of an edge to another vertex.
  More precisely, if we label each vertex with an integer $1$ through $|V|$
  and each edge with an integer $1$ through $|E|$, then
  vertex $i$ places its four corner pawns at positions
  $(i,3i),(i,3i+1),(4|V|+i,3i),(4|V|+i,3i+1)$;
  and edge $k$ connecting vertices $i$ and $j$
  adds pawns at the locations $(|V|+k,3i),(|V|+k,3j)$,
  in the bottom rows of vertex $i$'s and vertex $j$'s rectangles respectively.
  Thus every column has zero or two pawns,
  and each row has between zero and five pawns.
  For the start vertex $s$ with incident edge \(k_s\), we add a single rook at \((k_s, 1)\).
  Similarly for the destination vertex $t$ with incident edge \(k_t\), we add a single pawn at \((k_t, 0)\).

  In this construction,
  none of the pawns can make any captures,
  so only the rook can ever move.
  The rook can only enter or exit a vertex via its three edge pawns,
  and thus can enter a vertex at most once: entering, exiting, and
  entering again would prevent ever exiting again in a maximum-degree-$3$
  graph, preventing us from getting to $t$ (which itself has degree $1$
  so it cannot ever be exited).
  In order to visit all the pawns,
  the rook must therefore enter and exit each vertex other than \(s\) and \(t\) exactly once.
  Once the rook enters a vertex
  it must therefore visit all seven pawns of the vertex.
  By going clockwise or counterclockwise around the vertex rectangle,
  the rook can choose to leave along either of the two other edges
  incident to the vertex.
  Thus the rook capturing all pawns from its starting location in $s$
  if and only if there is a Hamiltonian path from $s$ to~$t$.

  \begin{figure}
	\centering
	\includegraphics[scale=0.5]{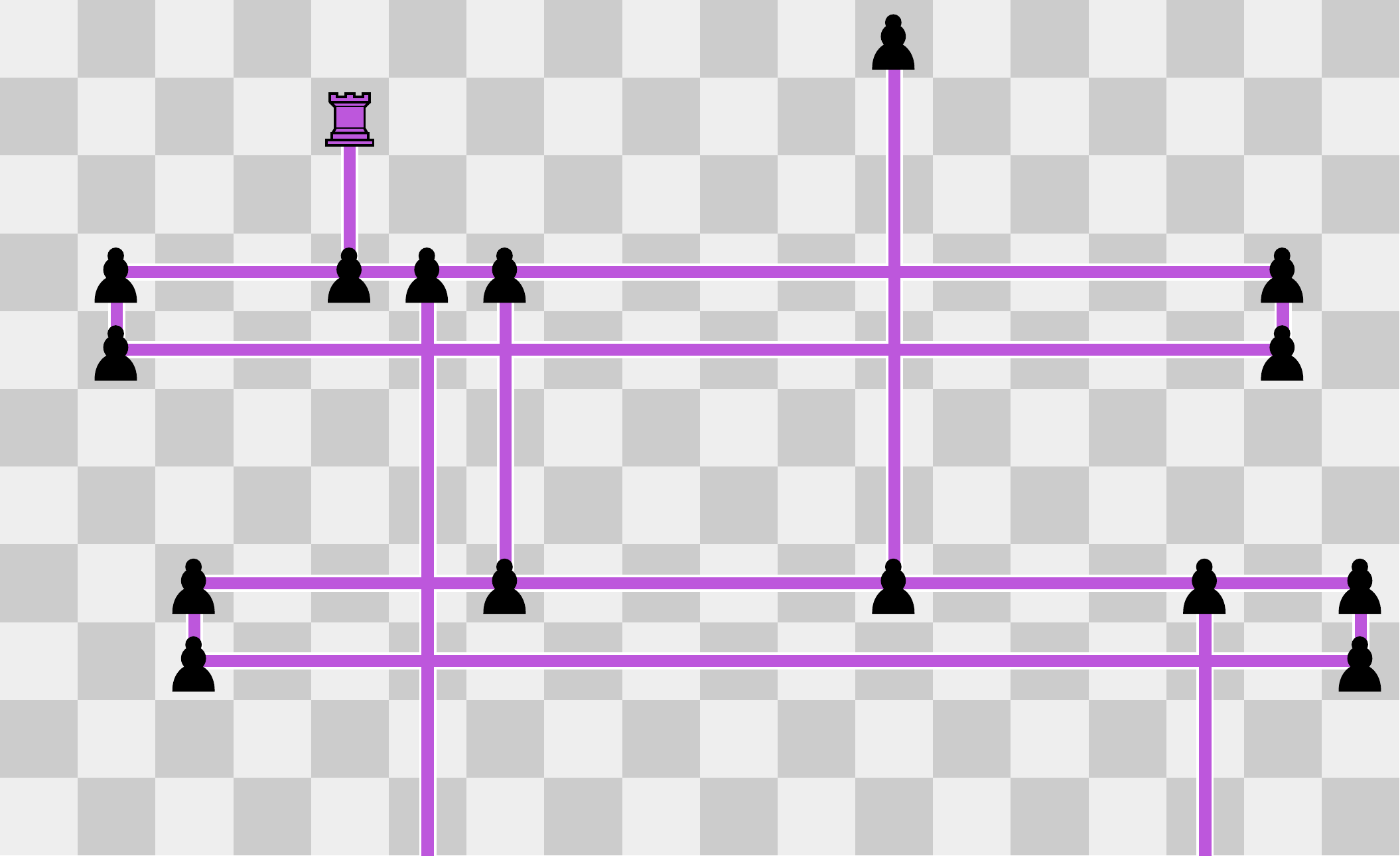}
    \caption{Placing pawns and one rook to simulate Hamiltonian Path in a maximum-degree-$3$ graph with specified start and end vertices.
      Two vertices connected by an edge are shown, as well as the start and destination vertices.}
	\label{fig:PawnRookNP}
  \end{figure}
\end{proof}

\begin{corollary}\label{cor:oneRookKnightKing}
  For any $T \in \{\king,\knight\}$,
  \solochessn{\{\onerook, T\}} and \solochessn{\{\rook, T\}} are NP-hard.
\end{corollary}
\begin{proof}
  Scale the construction from Theorem~\ref{thm:oneRook}
  by a factor of $4$ in both dimensions,
  and replace each pawn with a piece of type~$T$.
  This scaling prevents kings or knights from making captures,
  without affecting rook moves.
\end{proof}

\begin{corollary}\label{cor:oneRookBishop}\label{cor:oneQueen}\label{cor:oneQueenBishop}
  For any $S \in \{\rook,\queen\}$ and
  $T \in \{\pawn, \king, \knight, \bishop\}$.
  \solochessn{\{S^1, T\}} and \solochessn{\{S,T\}} are NP-hard.
\end{corollary}
\begin{proof}
  Scale the construction from Theorem~\ref{thm:oneRook} by \(h+1\)
  in the $x$ direction, where $h$ is the height of the original construction.
  (In other words, add $h$ empty columns between every consecutive
  pair of columns of pieces.)
  This scaling spaces out the pieces far enough so that
  no diagonal captures are possible.
  If $S = \queen$, replacing the rook with a queen
  does not add any additional diagonal moves.
  Finally we replace each pawn with a piece of type~$T$.
  If $T \in \{\pawn,\bishop\}$, these pieces cannot move
  because there are no diagonal moves.
  If $T \in \{\king,\knight\}$, we scale by an additional factor of $4$
  in both dimensions (as in Corollary~\ref{cor:oneRookKnightKing})
  to guarantee these pieces have no moves.
\end{proof}

\begin{corollary}\label{cor:oneBishop}
  For any $T \in \{\pawn, \king, \knight\}$,
  \solochessn{\{\onebishop, T\}} and \solochessn{\{\bishop, T\}}
  are NP-hard.
\end{corollary}
\begin{proof}
  Rotate the construction from Theorem~\ref{thm:oneRook}
  by \(45^\circ\) and scale by $\sqrt 2$; see Figure~\ref{fig:BishopPawnNP}.
  This transformation turns rook moves into bishop moves:
  vertices that were orthogonally adjacent are now diagonally adjacent.
  Replace the rook with a bishop,
  and replace each pawn with a piece of type~$T$.
  For $T \in \{\king, \knight\}$, we scale by an additional factor of $4$
  in both dimensions (as in Corollary~\ref{cor:oneRookKnightKing})
  to guarantee that these pieces have no moves.
\end{proof}

\begin{figure}
	\centering
	\includegraphics[scale=0.5]{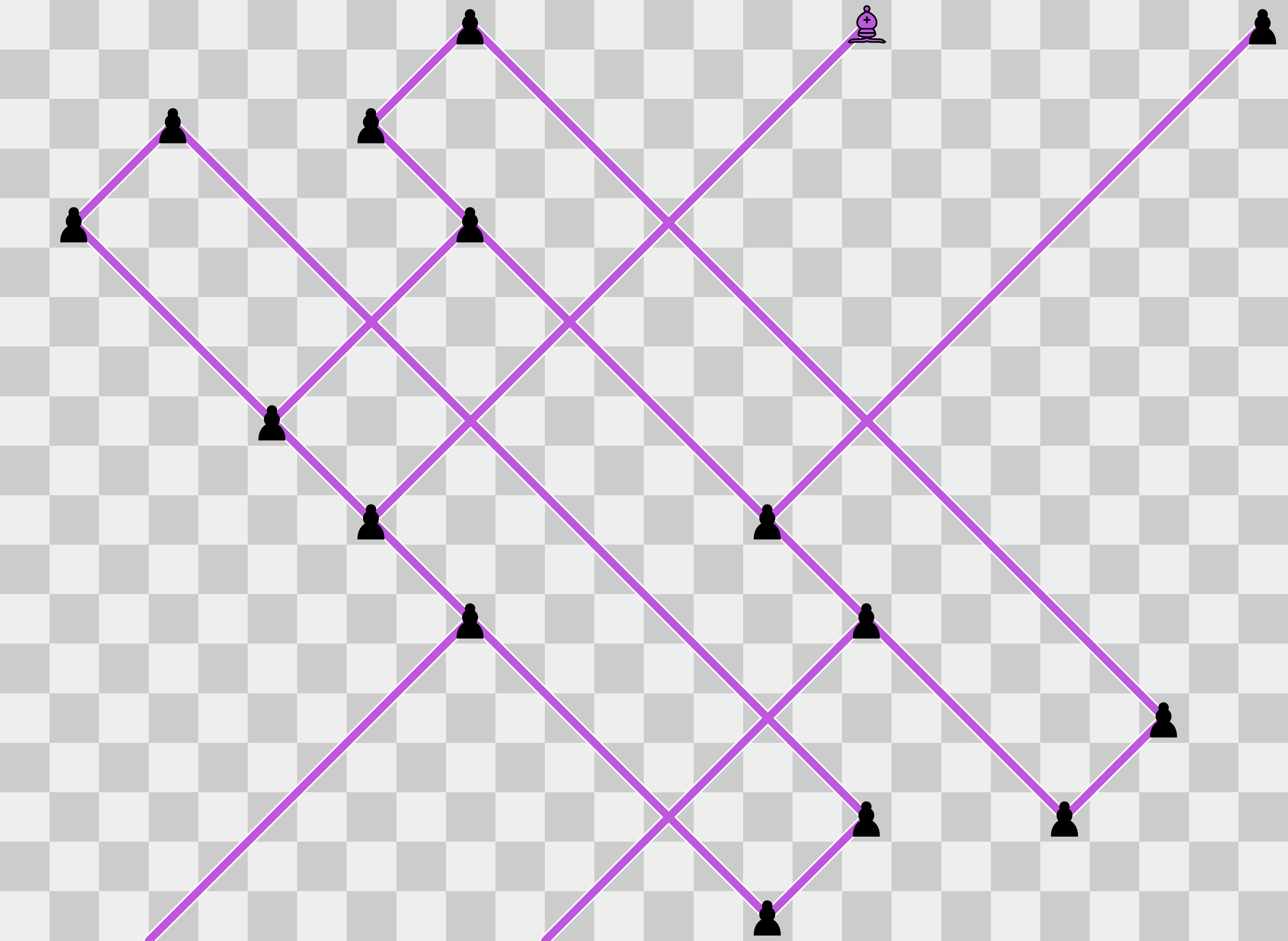}
	\caption{Placing pawns and one bishop to simulate Hamiltonian Path in a maximum-degree-$3$ graph with specified start and end vertices, by rotating the construction in Figure~\ref{fig:PawnRookNP}.}
	\label{fig:BishopPawnNP}
\end{figure}
	
\begin{corollary}\label{cor:oneBishopRook}\label{cor:oneQueenRook}
  For any $T \in \{\bishop, \queen\}$,
	\solochessn{\{T^1, \rook\}} and \solochessn{\{T, \rook\}} are NP-hard.
\end{corollary}
\begin{proof}
  First scale the construction from Theorem~\ref{thm:oneRook} by \(h+1\)
  in the $x$ direction, as in Corollary~\ref{cor:oneRookKnightKing}.
  Then rotate by $45^\circ$ and scale by~$\sqrt 2$,
  as in Corollary~\ref{cor:oneBishop}.
  The initial scaling eliminates diagonal alignments
  in the original construction,
  thus preventing pieces from aligning orthogonally in the rotated version.
  Then replace the rook with a piece of type~$T$,
  and replace each pawn with a rook.
\end{proof}

The next reduction is also from Hamiltonian Path,
but does not use the same framework as the previous reductions. 

\begin{theorem} \label{thm:king-pawn}
	\solochessn{\{\oneking,\pawn\}} and
  \solochessn{\{\king,\pawn\}} are NP-hard.
\end{theorem}
\begin{proof}
The reduction is from Hamiltonian Path in maximum-degree-3 grid graphs with a specified start vertex $s$ (Lemma~\ref{lem:start vertex}).
Figure~\ref{fig:1KingPawn} shows an example of the construction.
First, we rotate the given grid graph by $45^\circ$ and scale it by
$3 \sqrt 2$, placing pawns at the vertices and along the edges.
Pawns at vertices are drawn blue.
Adjacent vertex pawns are three spaces apart diagonally,
and have a diagonal chain of two (black) pawns between them.
All of these pawns forming the grid graph are placed on the light squares
of the board.

Assume pawns capture upward.
Now, for each vertex with at least one upward incident edge,
we place a green pawn on an adjacent dark square:
if there are two upward incident edges, then we place it above the vertex,
and otherwise we place it below the vertex.
Note that vertices with only downward incident edges do not get a green pawn;
in this case, we color the vertex pawn green.
The result is that every vertex has exactly one green pawn,
which has no legal captures, while all other pawns have legal captures.
All of the other pawns (black or blue in Figure~\ref{fig:1KingPawn})
have at least one legal capture.
The king replaces the green pawn on the starting vertex $s$ of the
Hamiltonian Path problem (the bottommost vertex in Figure~\ref{fig:1KingPawn}).

\begin{figure}
	\centering
	\includegraphics[scale=0.5]{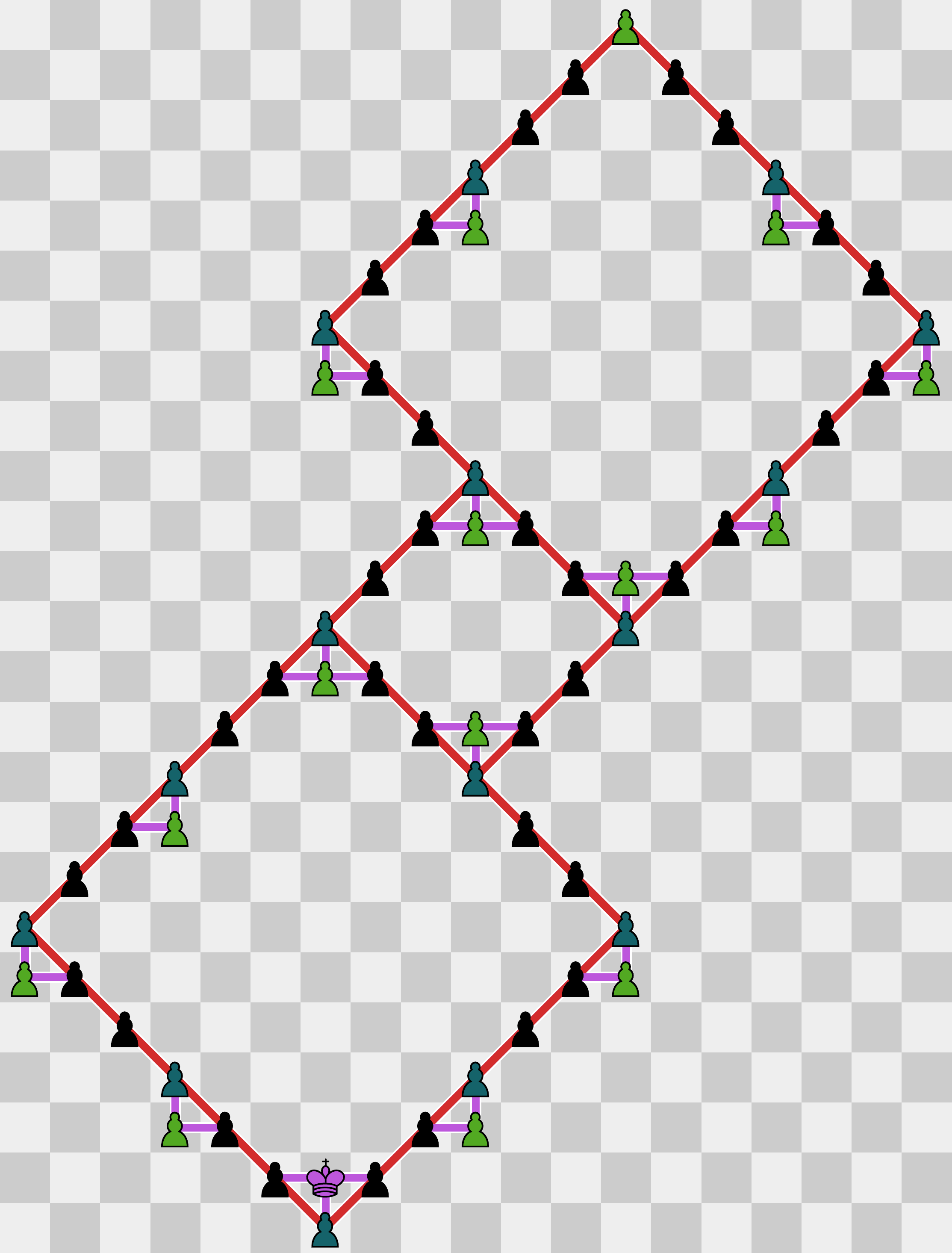}
  \caption{Placing pawns and one king to represent Hamiltonian Path in a maximum-degree-$3$ grid graph with a specified start vertex. Blue pawns are at grid-graph vertices, while green pawns have no capture move.}
	\label{fig:1KingPawn}
\end{figure}

If there is a Hamiltonian path, then we claim that
there is a valid capture sequence that leaves only the king at the end.
The king will capture along the Hamiltonian path,
making sure to divert and capture the green pawn at each vertex.
Before the king moves, though,
any pawns on squares which are not part of the Hamiltonian path
capture a pawn above them, starting from the bottommost pawns.
These pawns always have legal captures because, by our construction,
every pawn either has a legal capture or is a green pawn
that is part of the Hamiltonian path.
After these captures happen,
the only remaining pawns are those on the Hamiltonian path,
so the king can simply walk along that path, taking all the pawns.

Conversely, if there is a valid solution to the Solo Chess problem,
then we claim that there must exist a Hamiltonian path
in the underlying grid graph.
Because the green pawns can never move,
the king must at some point capture every green pawn.
Thus the king's path starts at the king's initial position,
passes through pawns, never captures the same square twice,
and captures every green pawn.
Because the graph has maximum degree $3$,
the king can visit each vertex at most once,
and because every vertex has a green pawn adjacent to it,
the king's path must be able to visit each vertex at least once.
Thus the king's path provides a Hamiltonian path in the graph.
(This argument works even without the $\oneking$ restriction because,
if the king gets captured before reaching all the green pawns,
the puzzle cannot be solved.)
\end{proof}

At this point, we have completed our proof that \solochessn{S} is NP-complete
for any two standard Chess pieces:

\begin{corollary}
  \solochessn{S} is NP-complete for any set $S$ of two distinct
  standard Chess pieces.
\end{corollary}

\begin{proof}
  Corollary~\ref{cor:oneQueen} and~\ref{cor:oneQueenBishop}
  together cover all cases where $\rook \in S$ or $\queen \in S$,
  leaving $S \subseteq \{\pawn, \king, \knight, \bishop\}$.
  Corollary~\ref{cor:oneBishop} covers all remaining cases
  where $\bishop \in S$,
  leaving $S \subseteq \{\pawn, \king, \knight\}$.
  Theorem~\ref{thm:short range} covers all remaining cases
  where $\knight \in S$, leaving $S \subseteq \{\pawn, \king\}$.
  Theorem~\ref{thm:king-pawn} covers the final case $S = \{\pawn, \king\}$.
\end{proof}

\subsection{SAT Reductions}
\label{SAT Reductions}

Next we turn to the uncapturable restriction for some of the piece types
not covered by previous reductions.
The reductions in this section are from a special case of 3SAT%
\footnote{By \defn{3SAT}, we mean CNF Satisfiability with \emph{at most}
  three variables per clause, rather than \emph{exactly} three variables
  per clause (E3SAT).}
with at most two occurrences of each literal,
which was shown to be NP-hard by Tovey \cite[Theorem~2.1]{tovey3sat3}.

It is convenient here to reduce from a planar version of 3SAT.
De~Berg and Khosravi \cite[Theorem~1]{sided3sat}
prove NP-hardness of \defn{Planar Monotone 3SAT}.
In this variation of 3SAT, the graph with
a vertex for each clause, a vertex for each variable,
edges between each clause and the variables it contains,
and a Hamiltonian cycle passing through all the variables,
must have a planar embedding.
Furthermore, in this embedding, all clauses containing positive literals
must be placed inside the Hamiltonian cycle, while all clauses containing
negative literals must be placed outside it;
in particular, every clause either consists entirely of positive literals
or consists entirely of negative literals.
(A more precise name for this problem is
``Sided Var-Linked Planar Monotone 3SAT'' \cite{ivan-thesis}.)
Equivalently, one can imagine arranging the variables along a line
in the plane, with all positive clauses (and their edges)
on one side of the line, and all negative clauses on the other side.

We also want the condition that each literal occurs at most twice.
In Appendix~\ref{app:sided-3sat-(1,2)}, we show that the combined problem ---
Planar Monotone 3SAT with at most two occurrences of each literal --- remains
NP-hard.

\begin{theorem} \label{thm:king-knight}
  For any $T \in \{\pawn,\king\}$,
	\solochessn{\{T^1,\knight\}} is NP-hard.
\end{theorem}

\begin{proof}
  We reduce from Planar Monotone 3SAT with at most two occurrences
  of each literal (Lemma~\ref{lem:sided-3sat-(1,2)}).
  Refer to Figure~\ref{fig:1KingKnight}.

  For each variable \(x_i\), we construct a variable gadget
  consisting of two pawn-traversable paths.
  Each path contains two \defn{literal knights} (drawn green)
  corresponding to literals for that variable:
  the literal knights on the left path correspond to the positive literal \(x_i\),
  while the literal knights on the right path correspond to the negative literal \(\neg x_i\).
  These literal knights are connected together by noncrossing paths of knights corresponding to the 3SAT clauses.
  The stipulation that each literal occurs at most twice ensures that we have enough literal knights to construct the 3SAT instance.

  Assume pawns capture upward.
  The lone pawn or king must traverse the board from bottom to top,
  visiting each variable gadget in turn.
  At each variable gadget, it is presented with a choice of two paths,
  allowing it to visit either the positive literal knights
  or the negative literal knights for that variable, but not both.
  (A king could go up one literal path and down the other literal path,
  but then it would get stuck, unable to reach a final knight at the top
  of the construction.)
  In order to capture the knights used in clauses,
  at least one literal knight from each clause must be visited by the
  pawn or king.
  All other knights, including literal knights not used in a clause,
  are connected to both sides of the gadget,
  so that they may be captured regardless of which path is taken.
  Thus the Solo Chess instance can be solved if and only if
  the 3SAT instance is satisfiable.
\end{proof}

\begin{figure}
  \centering
  $\vcenter{\hbox{\includegraphics[scale=0.42]{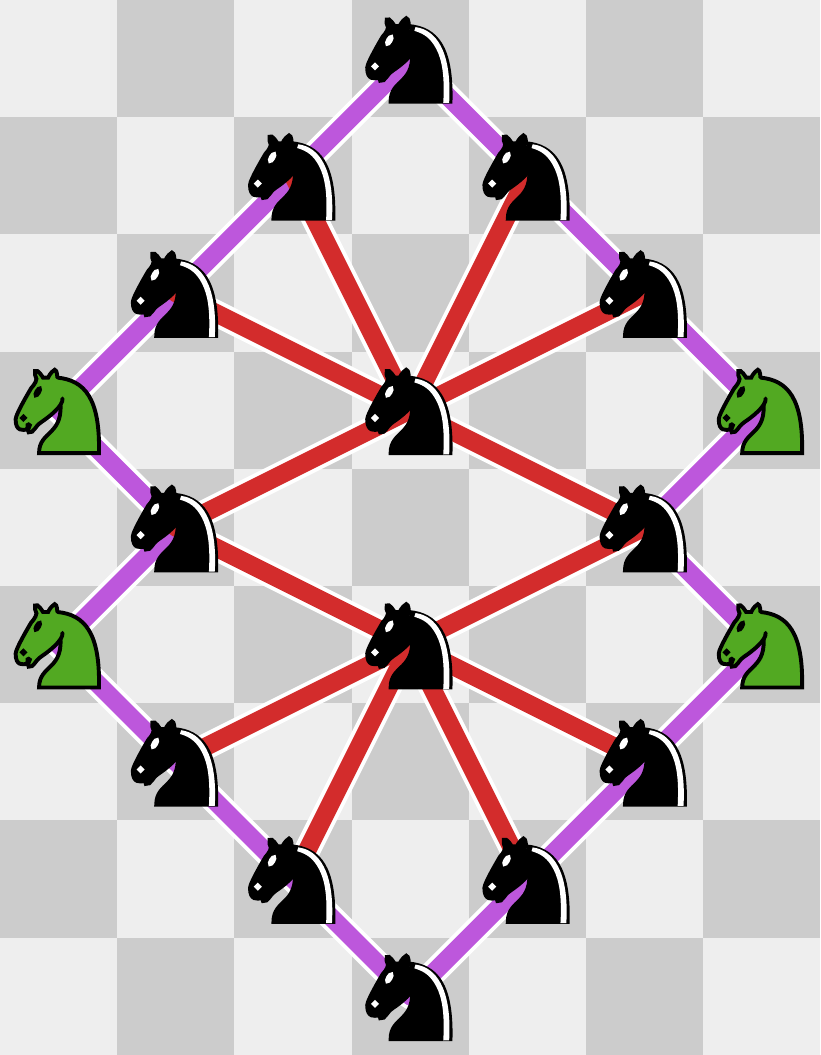}}}$\hfil
  $\vcenter{\hbox{\includegraphics[scale=0.42]{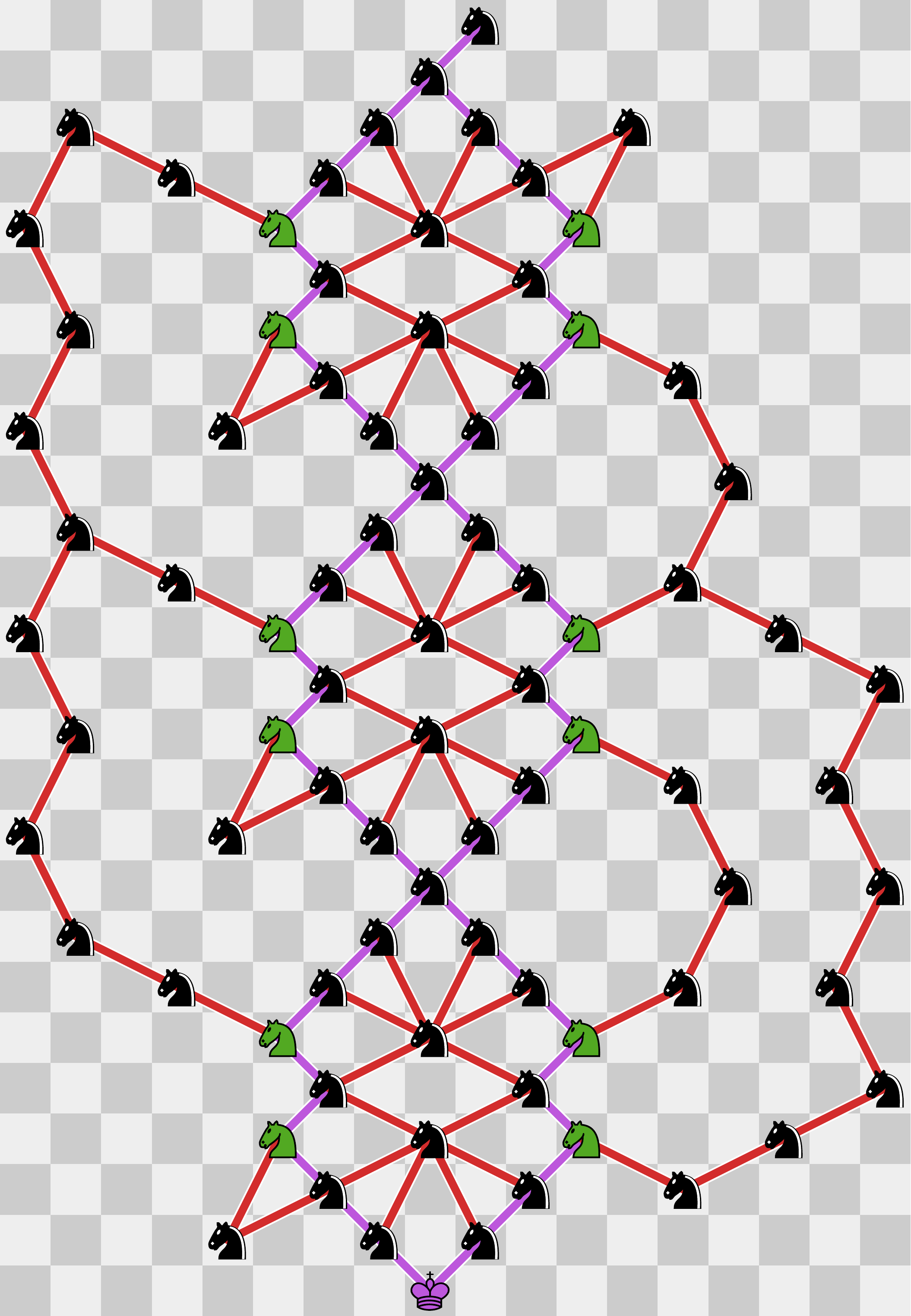}}}$
  \caption{
    Variable gadget (left) and example reduction output (right)
    for \solochessn{\{\oneking, \knight\}}.
    This instance corresponds to the formula \((x_1 \vee x_2 \vee x_3) \wedge (\neg x_1 \vee \neg x_2 \vee \neg x_3) \wedge (\neg x_2 \vee \neg x_3)\).
    At least one green literal knight must be visited in each clause.
  }
  \label{fig:1KingKnight}
\end{figure}

The other reductions in this section are similar; we just have to design a suitable variable gadget in each case.

\begin{theorem}
  \label{thm:king-bishop}
  \solochessn{\{\oneking,\bishop\}} is NP-hard.
\end{theorem}
\begin{proof}
  We reduce from Planar Monotone 3SAT with at most two occurrences
  of each literal, similar to Theorem~\ref{thm:king-knight}.
  Instead, we use the variable gadget shown in Figure~\ref{fig:1KingBishop}.
  We add bishops to connect literal bishops in each clause,
  or to connect unused literal bishops to both paths
  so that they may be captured regardless of which is chosen.
\end{proof}

\begin{figure}
	\centering
	\includegraphics[scale=0.42]{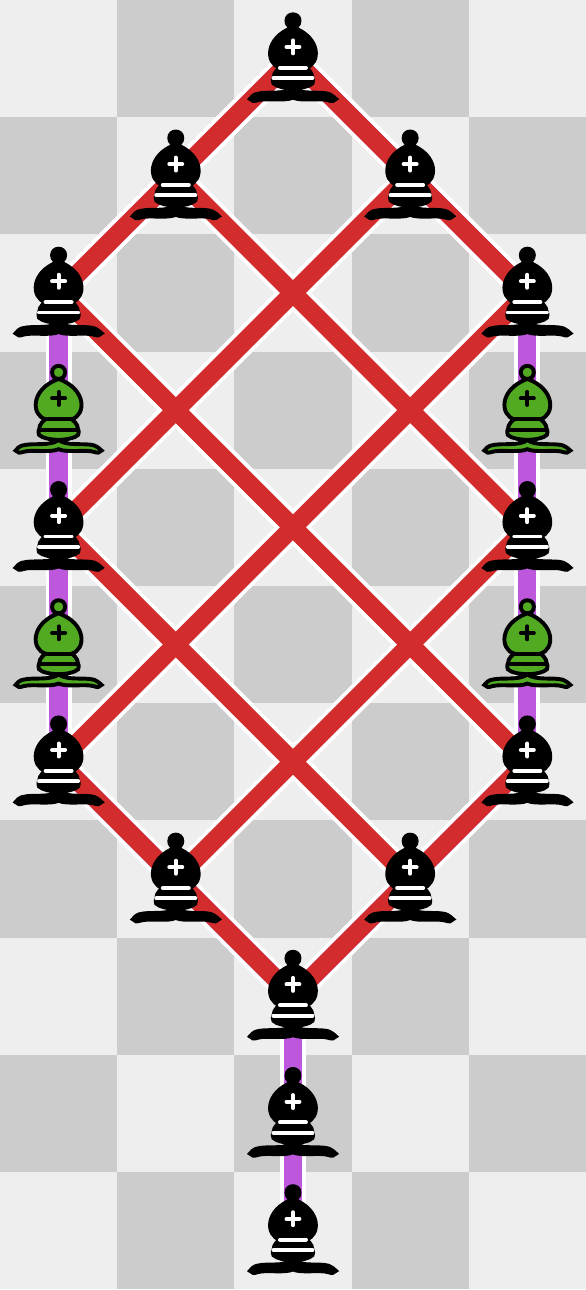}
	\caption{Variable gadget for \solochessn{\{\oneking, \bishop\}}.}
	\label{fig:1KingBishop}
\end{figure}

\begin{theorem}
  \label{thm:knight-queen}
  For any $T \in \{\bishop, \rook, \queen\}$,
  \solochessn{\{\oneknight, T\}} is NP-hard.
\end{theorem}
\begin{proof}
  We reduce from Planar Monotone 3SAT with at most two occurrences
  of each literal, similar to Theorem~\ref{thm:king-knight}.
  Instead, we use the variable gadget shown in Figure~\ref{fig:1KnightQueen}
  (for the case $T = \queen$).
  Note that all non-literal pieces are connected to both paths
  even if the queens in the diagram are replaced by rooks or bishops.
  We add queens/rooks/bishops to connect literal queens/rooks/bishops
  in each clause, or to connect unused literal queens/rooks/bishops
  to both paths.
\end{proof}

\begin{figure}
	\centering
	\includegraphics[scale=0.42]{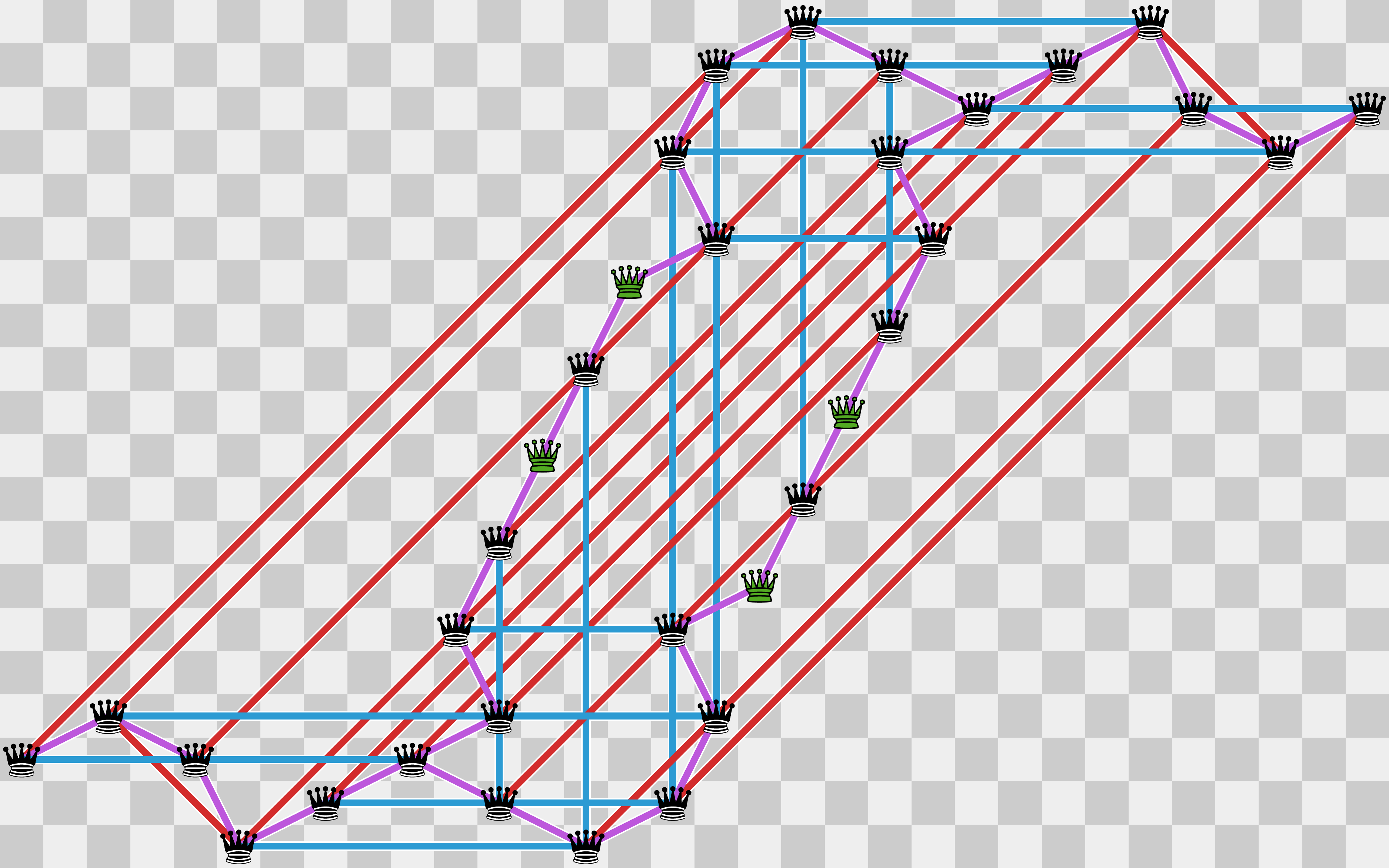}
	\caption{Variable gadget for \solochessn{\{\oneknight, \queen\}}.}
	\label{fig:1KnightQueen}
\end{figure}

\section{Open Problems}
\label{Open Problems}

Our results in Table~\ref{tab:results} leave open the complexity
of two cases: $\{\oneking,\rook\}$ and $\{\onepawn,\rook\}$.
We suspect that both of these problems can be solved in polynomial time,
essentially because a king or pawn cannot slip by a rook,
but it remains to generalize the algorithm in Section~\ref{sec:pairs-easy}.
Similarly, it is open whether the algorithm can be generalized to non-symmetric pieces.

From the prior paper \cite{aravind2022chess},
the complexities of \solochessn{\{\knight\},O(1)} and
\solochessn{\{\king\},O(1)} are still open.
It may help to show hardness for the more nonblocking piece types
on a graph, possibly constrained to have maximum degree $8$ or
to be a grid graph.

Finally, although Solo Chess puzzles are not designed to ensure a unique solution,
it is interesting to determine whether the problem is ASP-complete
and whether counting the number of solutions is \#P-complete.
Some, but not all, of our reductions are parsimonious.

\section*{Acknowledgments}

This work was initiated during extended problem solving sessions
with the participants of the MIT class on
Algorithmic Lower Bounds: Fun with Hardness Proofs (6.892)
taught by Erik Demaine in Spring 2019.
We thank the other participants for their insights and contributions.
In particular, we thank Dylan Hendrickson for helpful discussions
around algorithms for one piece type.

Most figures of this paper were drawn using SVG Tiler
[\url{https://github.com/edemaine/svgtiler}].
Chess piece images are based on Wikipedia's
\url{https://commons.wikimedia.org/wiki/Standard_chess_diagram},
drawn by Colin M.L. Burnett and licensed under a BSD License.

\bibliographystyle{alpha}
\bibliography{bibliography}

\appendix
\section{Hamiltonian Path in Maximum-Degree-3 Grid Graphs with Specified Start/End Vertices}
\label{app:start vertex}

Itai, Papadimitriou, and Szwarcfiter \cite{GridHamPath} prove NP-hardness of
deciding whether a grid graph has a Hamiltonian path \emph{with specified
start and end vertices}.%
\footnote{They also describe how to reduce this problem to
  deciding whether a graph has a Hamiltonian path
  (with no specified start/end vertices), but their reduction
  (attaching a degree-$1$ vertex to each of the specified start and
  end vertices) does not obviously preserve being a grid graph.}
Papadimitriou and Vazirani \cite{Degree3GridHamPath} prove NP-hardness of
deciding whether a \emph{maximum-degree-3} grid graph has a Hamiltonian path
(with no specified start/end vertices).
Neither result is exactly what we need in this paper:

\begin{lemma} \label{lem:start vertex}
  It is NP-hard to decide whether a maximum-degree-3 grid graph has a
  Hamiltonian path that
  \begin{enumerate}
  \item starts at a specified start vertex~$s$, which is degree $1$
    (but without a specified end vertex); or
  \item starts at a specified start vertex~$s$
    and ends at a specified end vertex~$t$,
    both of which are degree~$1$.
  \end{enumerate}
\end{lemma}

\begin{proof}
  We modify the proof of Papadimitriou and Vazirani \cite{Degree3GridHamPath}.
  Their proof reduces from Hamiltonian Circuit in
  a planar \emph{directed} graph $G_1$ where each vertex has
  either in-degree $2$ and out-degree $1$ or in-degree $1$ and out-degree~$2$.

  Their first modification to $G_1$ \cite[Figure~2]{Degree3GridHamPath}
  forms an \emph{undirected} graph $G_2$ such that
  $G_1$ has a Hamiltonian cycle if and only if $G_2$ has a Hamiltonian path.
  Part of this modification \cite[Figure~2(b)]{Degree3GridHamPath}
  replaces one vertex $v_1$ of $G_1$ with a gadget of four vertices
  that includes two degree-$1$ vertices $\textsf{vin}_1'$ and
  $\textsf{vout}_1$.
  Clearly if $G_2$ has a Hamiltonian path, then it must start and end
  at $\textsf{vin}_1'$ and $\textsf{vout}_1$.

  Next their proof forms a maximum-degree-$3$ grid graph $G'_4$ such that
  $G'_4$ has a Hamiltonian path if and only if $G_2$ has a Hamiltonian path.
  $G'_4$ is essentially a grid drawing of $G_2$
  (which turns out to be bipartite), expanded by a constant factor,
  and replacing each vertex and edge by a thickened gadget.
  The degree-1 vertices of $G_2$, $\textsf{vin}_1'$ and $\textsf{vout}_1$,
  are each mapped in $G_4$ to a ``dumbbell''
  (two length-$8$ cycles connected via a length-$6$ path)
  attached to a single ``tentacle''
  (a $2 \times n$ rectangle with length-$8$ cycles at turns)
  via a ``pin connection'' (single adjacency);
  see Figure~\ref{fig:dumbbell-pin}.
  Because the dumbbell is connected to the rest of the graph via only
  a single edge (the pin connection),
  any Hamiltonian path must start or end within each such dumbbell.
  In particular, we can choose a particular start or end vertex within
  the dumbbell to be either vertex adjacent to the far end
  (relative to the pin connection) of the path between the two cycles;
  we label such a vertex $v$ in Figure~\ref{fig:dumbbell-pin}.
  By \cite[Lemma]{Degree3GridHamPath} or Figure~\ref{fig:dumbbell-deg1-ham},
  if we declare this vertex $v$ to be the specified start or end vertex,
  then we preserve the existence of a Hamiltonian path in~$G_4$.
  This vertex $v$ has degree~$2$, and has a neighboring grid point
  (below $v$ in Figure~\ref{fig:dumbbell-pin})
  with no neighboring vertices, so we can add a degree-$1$ vertex $\ell$
  at that point as shown in Figure~\ref{fig:dumbbell-deg1},
  increasing $v$'s degree from $2$ to~$3$ (preserving maximum-degree-$3$),
  and then the degree-$1$ vertex $\ell$ must be the start or end of any
  Hamiltonian path.
  Figure~\ref{fig:dumbbell-deg1-ham} shows the local configuration any
  Hamiltonian path must have (in particular verifying the preservation of
  the existence of a Hamiltonian path).
  We can declare the vertex from $\textsf{vin}_1'$ to be the start vertex,
  and optionally declare the vertex from $\textsf{vout}_1$ to be the end vertex,
  to prove the two variations NP-hard.
\end{proof}

\begin{figure}
  \centering
  \subcaptionbox{\label{fig:dumbbell-pin}
    A dumbbell with a pin connection,
    based on \cite[Figure~12(b)]{Degree3GridHamPath}}
    {\begin{overpic}[scale=0.85]{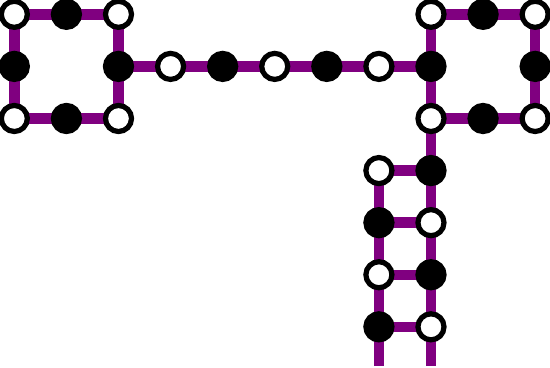}
      \put(27.5,45){\makebox(0,0)[c]{$v$}}
    \end{overpic}}\hfill
  \subcaptionbox{\label{fig:dumbbell-deg1}Adding a degree-$1$ vertex}
    {\begin{overpic}[scale=0.85]{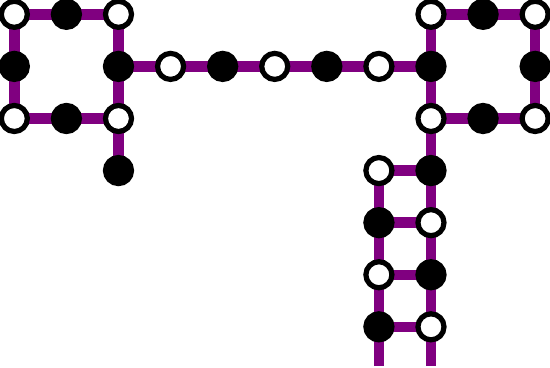}
      \put(27.5,45){\makebox(0,0)[c]{$v$}}
      \put(27.5,35.8){\makebox(0,0)[c]{$\ell$}}
    \end{overpic}}\hfill
  \subcaptionbox{\label{fig:dumbbell-deg1-ham}Resulting Hamiltonian path}
    {\begin{overpic}[scale=0.85]{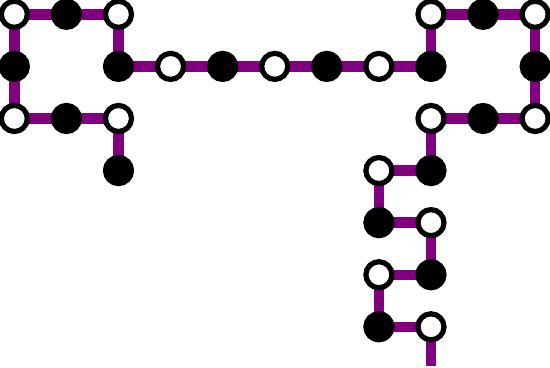}
      \put(27.5,45){\makebox(0,0)[c]{$v$}}
      \put(27.5,35.8){\makebox(0,0)[c]{$\ell$}}
    \end{overpic}}
  \caption{The hardness reduction to Hamiltonian path in
    maximum-degree-$3$ grid graphs from \cite{Degree3GridHamPath}
    has two copies of the gadget in (a).
    This graph is Hamiltonian if and only if the modification in (b) is,
    which forces the Hamiltonian path to look like (c),
    in particular starting or ending at~$\ell$.}
  \label{fig:max-deg-3 fix}
\end{figure}

\section{Sided Var-Linked Planar Monotone 3SAT with Restricted Variable Occurrences}
\label{app:sided-3sat-(1,2)}

\begin{lemma} \label{lem:sided-3sat-(1,2)}
  Sided Var-Linked Planar Monotone 3SAT is NP-hard,
  even when each literal occurs at most twice and each variable occurs at most three times.
\end{lemma}
Knuth and Raghunathan \cite{rectilinearsat} observe that any instance of Var-Linked Planar 3SAT has a \emph{rectilinear} layout.  That is, the clauses and variables can be drawn as horizontal line segments in the plane with vertical line segments connecting incident variables to clauses, such that no line segments intersect each other otherwise and all variables lie on the same horizontal line.  Thus this result also extends to the version of the problem where such a rectilinear layout is provided.

\begin{proof}
  The reduction is from Sided Var-Linked Planar Monotone 3SAT.
  Refer to Figure~\ref{fig:sided-3sat-(1,2)}.
  We show that each variable can be replaced by a set of new variables and clauses to form an equisatisfiable instance of Sided Var-Linked Planar Monotone 3SAT, such that the new variables each have at most three occurrences, at most two of which have the same sign.
  In the following we assume each variable has at least one occurrence of each sign; any variable which doesn't can be deleted without affecting satisfiability.

  Let \(x\) be a variable with \(n\) negative occurrences in clauses \(N_1, \dots, N_n\) and  \(m\) positive occurrences in clauses \(P_1, \dots, P_m\).
  We replace \(x\) by two sequences of variables \(x_1, \dots, x_{n+m}\) and \(y_1, \dots, y_{n+m-1}\).
  For each \(1 \le i \le m + n - 1\) we add a new positive clause \(x_i \vee y_i\) and a new negative clause \(\neg y_i \vee \neg x_{i+1}\).
  Finally for each \(1 \le k \le n\) we replace each occurrence of \(\neg x\) in clause \(N_k\) with an occurrence of \(\neg x_k\), and for each \(1 \le j \le m\) we replace each occurrence of \(x\) in clause \(P_j\) with an occurrence of \(x_{n+j}\).
  This completes the construction.  Each new variable \(x_i\) or \(y_i\) occurs at most three times and at most twice with the same sign.  It can be seen from Figure~\ref{fig:sided-3sat-(1,2)} that this construction preserves the sided planarity property.  We must show that the resulting instance is equisatisfiable with the original.

  In one direction, let \(v\) be a satisfying assignment (i.e. a mapping from variables to \(\{0, 1\}\)) for the original instance.
  Define a new assignment \(w\) by \(w(x_i) = v(x)\) and \(w(y_i) = \neg v(x)\) for each original variable \(x\).
  Then \(w\) is a satisfying assignment for the new instance.

  In the other direction, let \(w\) be a satisfying assignment for the new instance.
  Define an assignment \(v\) over the original variables \(x\) by \(v(x) = w(x_n)\) where \(n\) is the number of negative occurrences of \(x\).
  We claim that \(v\) is a satisfying assignment for the original instance.
  In order to do this it suffices to show for each variable \(x\) with \(n\) negative and \(m\) positive occurrences,
  that
  \begin{equation}
    \label{eqn:wvsat}
    \begin{aligned}
      w(x_k) &\ge v(x) \\
      w(x_{n+j}) &\le v(x)
    \end{aligned}
  \end{equation}
  for \(1 \le k \le n\) and \(1 \le j \le m\).
  For then any clause \(N_k\) or \(P_j\) which was satisfied by having \(w(x_k) = 0\) or \(w(x_{n+j}) = 1\)
  is also satisfied by the value of \(v(x)\).

  The clauses \(x_i \vee y_i\) and \(\neg y_i \vee \neg x_{i+1}\) together require \(x_{i+1} \to x_i\) for \(1 \le i \le m+n-1\).  Thus \(w\) is weakly monotonically decreasing on \(x_1, \dots, x_{n+m}\).  Since \(v(x) = w(x_n)\) this immediately yields (\ref{eqn:wvsat}).
  Thus the new instance of Sided Var-Linked Planar Monotone 3SAT is equisatisfiable with the original.
\end{proof}

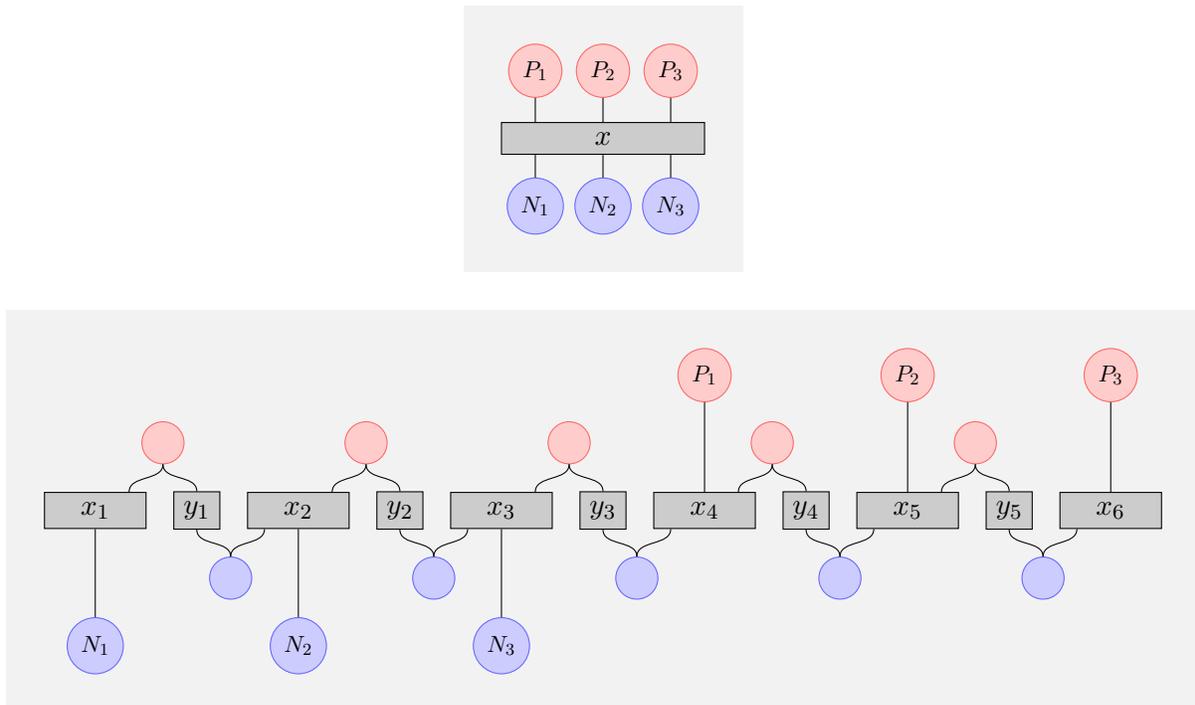
\begin{figure}
  \centering
  \begin{tikzpicture}
    [
      x=0.9cm,
      y=0.9cm,
      var/.style={rectangle,draw=black,fill=black!20},
      posc/.style={circle,minimum size=0.7cm,scale=0.8,draw=red!60,fill=red!20},
      negc/.style={posc,draw=blue!60,fill=blue!20},
    ]\SetScales
    \node at (0, 0) [var, minimum width=3\scaledx] (x) {\(x\)};
    \node at (-1, 1)  [posc] (sp1) {\(P_1\)};
    \node at (0, 1)   [posc] (sp2) {\(P_2\)};
    \node at (1, 1)   [posc] (sp3) {\(P_3\)};
    \node at (-1, -1) [negc] (sn1) {\(N_1\)};
    \node at (0, -1)  [negc] (sn2) {\(N_2\)};
    \node at (1, -1)  [negc] (sn3) {\(N_3\)};
    \draw (sp1.south) -- (sp1 |- x.north);
    \draw (sp2.south) -- (sp2 |- x.north);
    \draw (sp3.south) -- (sp3 |- x.north);
    \draw (sn1.north) -- (sn1 |- x.south);
    \draw (sn2.north) -- (sn2 |- x.south);
    \draw (sn3.north) -- (sn3 |- x.south);

    \begin{scope}[shift={(0,-5.5)}]
      \node at (-6, 0) [var] (y1) {\(y_1\)};
      \node at (-3, 0) [var] (y2) {\(y_2\)};
      \node at (0, 0)  [var] (y3) {\(y_3\)};
      \node at (3, 0)  [var] (y4) {\(y_4\)};
      \node at (6, 0)  [var] (y5) {\(y_5\)};

      \node at (-7.5, 0) [var, minimum width=1.5\scaledx] (x1) {\(x_1\)};
      \node at (-4.5, 0) [var, minimum width=1.5\scaledx] (x2) {\(x_2\)};
      \node at (-1.5, 0) [var, minimum width=1.5\scaledx] (x3) {\(x_3\)};
      \node at (1.5, 0)  [var, minimum width=1.5\scaledx] (x4) {\(x_4\)};
      \node at (4.5, 0)  [var, minimum width=1.5\scaledx] (x5) {\(x_5\)};
      \node at (7.5, 0)  [var, minimum width=1.5\scaledx] (x6) {\(x_6\)};

      \node at (-5.5, -1) [negc] (i1) {};
      \node at (-2.5, -1) [negc] (i2) {};
      \node at (0.5, -1) [negc] (i3) {};
      \node at (3.5, -1) [negc] (i4) {};
      \node at (6.5, -1) [negc] (i5) {};

      \draw (i1.north) to [out=90, in=270] (y1.south);
      \draw (i1.north) to [out=90, in=270] ([shift={(-0.5\scaledx,0)}]x2.south);
      \draw (i2.north) to [out=90, in=270] (y2.south);
      \draw (i2.north) to [out=90, in=270] ([shift={(-0.5\scaledx,0)}]x3.south);
      \draw (i3.north) to [out=90, in=270] (y3.south);
      \draw (i3.north) to [out=90, in=270] ([shift={(-0.5\scaledx,0)}]x4.south);
      \draw (i4.north) to [out=90, in=270] (y4.south);
      \draw (i4.north) to [out=90, in=270] ([shift={(-0.5\scaledx,0)}]x5.south);
      \draw (i5.north) to [out=90, in=270] (y5.south);
      \draw (i5.north) to [out=90, in=270] ([shift={(-0.5\scaledx,0)}]x6.south);

      \node at (-6.5, 1) [posc] (h1) {};
      \node at (-3.5, 1) [posc] (h2) {};
      \node at (-0.5, 1) [posc] (h3) {};
      \node at (2.5, 1) [posc] (h4) {};
      \node at (5.5, 1) [posc] (h5) {};

      \draw (h1.south) to [out=270, in=90] (y1.north);
      \draw (h1.south) to [out=270, in=90] ([shift={(0.5\scaledx,0)}]x1.north);
      \draw (h2.south) to [out=270, in=90] (y2.north);
      \draw (h2.south) to [out=270, in=90] ([shift={(0.5\scaledx,0)}]x2.north);
      \draw (h3.south) to [out=270, in=90] (y3.north);
      \draw (h3.south) to [out=270, in=90] ([shift={(0.5\scaledx,0)}]x3.north);
      \draw (h4.south) to [out=270, in=90] (y4.north);
      \draw (h4.south) to [out=270, in=90] ([shift={(0.5\scaledx,0)}]x4.north);
      \draw (h5.south) to [out=270, in=90] (y5.north);
      \draw (h5.south) to [out=270, in=90] ([shift={(0.5\scaledx,0)}]x5.north);

      \node at (-7.5, -2) [negc] (n1) {\(N_1\)};
      \node at (-4.5, -2) [negc] (n2) {\(N_2\)};
      \node at (-1.5, -2) [negc] (n3) {\(N_3\)};
      \node at (1.5, 2)   [posc] (p1) {\(P_1\)};
      \node at (4.5, 2)   [posc] (p2) {\(P_2\)};
      \node at (7.5, 2)   [posc] (p3) {\(P_3\)};

      \draw (n1.north) -- (n1 |- x1.south);
      \draw (n2.north) -- (n2 |- x2.south);
      \draw (n3.north) -- (n3 |- x3.south);
      \draw (p1.south) -- (p1 |- x4.north);
      \draw (p2.south) -- (p2 |- x5.north);
      \draw (p3.south) -- (p3 |- x6.north);
    \end{scope}

    \begin{scope}[on background layer]
      \definecolor{bgcolor}{RGB}{221,170,51}
      \node (r1) [rectangle,fill=black!5,fit=(sp1)(sp3)(sn1)(sn3)(x),inner sep=0.5cm] {};
      \node (r2) [rectangle,fill=black!5,fit=(p3)(n1)(x1)(x6),inner sep=0.5cm] {};
    \end{scope}
  \end{tikzpicture}
  \caption{Transforming a variable in Sided Var-Linked Monotone Planar 3SAT to reduce the number of occurrences.  Clauses above the line of variables (red) are positive; clauses below the line of variables (blue) are negative.  Above: A variable with three positive occurrences and three negative occurrences.  Below: An equivalent collection of clauses and variables.  Each variable occurs at most three times and has at most two occurrences with the same sign.}
  \label{fig:sided-3sat-(1,2)}
\end{figure}

\end{document}